\documentclass[11pt]{article}             
\usepackage{osid}                         %

\usepackage{amsmath,amssymb,amscd,amsthm,mathrsfs}
\usepackage{graphicx}
\usepackage{color}

\usepackage[bookmarks=false,pdfstartview={FitH}]{hyperref}

\newcommand{\Adr}{\operatorname{Ad}}

\newcommand{\ad}{\mathrm{ad}}
\newcommand{\trace}{\operatorname{tr}}

\newcommand{\clo}[2]{\ensuremath{\overline{#2}^{#1}}{}}
\newcommand{\clos}[1]{\clo{s}{#1}}
\newcommand{\clou}[1]{\clo{u}{#1}}
\newcommand{\clon}[1]{\clo{1}{#1}}

\newcommand{\uu}{\mathfrak{u}}
\newcommand{\su}{\mathfrak{su}}

\newcommand{\fk}{\mathfrak{k}}

\newcommand{\bK}{\mathbf{K}}

\newcommand{\R}{\mathbb{R}}

\newcommand{\reach}{\mathfrak{reach}}

\newcommand{\hGA}{\mathbf{\Gamma}}

\newcommand{\hH}{\mathbf {H}}

\newcommand{\hU}{\mathbf {U}}

\newcommand{\thmskip}{\medskip\smallskip} 

\theoremstyle{plain}
\newtheorem{lemma}{Lemma}
\newtheorem{thm}{Theorem}
\newtheorem{corollary}{Corollary}
\newtheorem{proposition}{Proposition}

\theoremstyle{definition}
\newtheorem{defi}{Definition}

\theoremstyle{remark}
\newtheorem{remark}{Remark}

\newcounter{app}
\renewcommand*{\theapp}{\Alph{app}}

\newcommand*\app[1]{%
  \refstepcounter{app}\label{app_#1}\hypertarget{app:#1}{\theapp}}
\newcommand*\appref[1]{\hyperlink{app:#1}{\ref*{app_#1}}}

\title{Reachability in Infinite Dimensional Unital Open Quantum Systems with Switchable GKS-Lindblad Generators} 
\author{Frederik vom Ende
	\\[1mm]{\footnotesize\it Technische Universit{\"a}t M{\"u}nchen, 
	Dept. Chem., Lichtenbergstra{\ss}e 4, 85747 Garching and\\
   Munich Centre for Quantum Science and Technology (MCQST),  Schellingstra{\ss}e~4, 80799~M{\"u}nchen, Germany \& {frederik.vom-ende@tum.de}}\\[2ex]
        Gunther Dirr
        \\[1mm]{\footnotesize\it Universit{\"a}t W{\"u}rzburg, 
        Institut f{\"u}r Mathematik, Emil-Fischer-Stra{\ss}e 40,\\
	97074 W{\"u}rzburg, Germany \& dirr@mathematik.uni-wuerzburg.de}\\[2ex] 
	Michael Keyl\\[1mm]{\footnotesize\it Freie Universit\"at Berlin, Dahlem Center for Complex Quantum Systems, Arnimallee 14,\\ 14195 Berlin, Germany \& {michael.keyl@tum.de}}\\[2ex]
Thomas Schulte-Herbr{\"u}ggen
	\\[1mm]{\footnotesize\it Technische Universit{\"a}t M{\"u}nchen, 
	Dept. Chem., Lichtenbergstra{\ss}e 4, 85747 Garching and\\
   Munich Centre for Quantum Science and Technology (MCQST),  Schellingstra{\ss}e~4,\\ 80799~M{\"u}nchen, Germany \& {tosh@tum.de}}
}

\AtBeginDocument{
  \label{CorrectFirstPageLabel}
  
}

\begin{document}

\maketitle
\begin{abstract}
In quantum systems theory one of the fundamental problems boils down to: given an initial state, which
final states can be reached by the dynamic system in question. Here we consider infinite dimensional 
open quantum dynamical systems following a unital Kossakowski-Lindblad master equation extended
by controls. More precisely, their time evolution shall be governed by an inevitable potentially unbounded
Hamiltonian drift term $H_0$, finitely many bounded control Hamiltonians $H_j$ allowing for (at least) 
piecewise constant control amplitudes $u_j(t)\in{\mathbb R}$ plus a bang-bang (i.e.\ on-off) switchable noise 
term $\hGA_V$ in Kossakowski-Lindblad form. Generalizing standard majorization results from finite to infinite 
dimensions, we show that such bilinear quantum control systems allow to approximately {\em reach
any target state majorized by the initial one} as up to now only has been known in finite dimensional
analogues.---The proof of the result is currently limited to the control Hamiltonians $ H_j$ being bounded and 
noise terms $\hGA_V$ with compact normal $V$.
\end{abstract}

\section{Introduction and Overview}

\subsection{Markovian Bilinear Quantum Control}\label{section_open_bil}

The Kossakowski-Lindblad equation~\cite{Koss72,Koss72b,Lind76,GKS76} plays a central role in quantum
dynamics since it characterizes the infinitesimal generators of the semigroup of 
all (invertible\footnote{Here invertibility only means invertible as linear map, not necessarily 
as quantum map.}) Markovian quantum maps. 

As in \cite{Wolf08a} a quantum map ({\sc cptp} map\footnote{ {\sc cptp} maps are linear 
\underline{c}ompletely \underline{p}ositive and \underline{t}race-\underline{p}reserving.}) 
is called 
{\em (time-dependent) Markovian}, if it is the solution of a (time-dependent) Markovian master equation
\begin{equation}\label{eq:master-map}
\dot F(t) = -\big(i\, \hH(t) + \hGA(t)\big) F(t), \quad F(0) = \mathbf 1\; ,
\end{equation}
where Markovianity and {\sc cptp} property are guaranteed by the Kossakowski-Lindblad form of $\hGA(t)$ in Eq.~\eqref{eq:GKSL}.
Then for finite dimensional  Hamiltonians and noise terms one can show that those 
Markovian quantum maps (including both, time-dependent and time-independent ones) are infinitesimal
divisible into products of exponentials of Kossakowski-Lindblad generators~\cite{Wolf08a} hence leading to 
{\em Lie semigroup} structure \cite{DHKS08}. 
In contrast, {\em non-Markovian} quantum maps (existing even arbitrarily close
to the identity map) are Kraus maps~\cite{Kraus83} that are {\em not} solutions of a Kossakowski-Lindblad master equation
and hence the set of {\em all}\/ invertible quantum maps (including Markovian and non-Markovian ones) has {\em no} Lie-semigroup structure
(details in~\cite{Wolf08a,DHKS08,OSID17}).

For most of the work, we focus on the corresponding induced system $\Sigma$ acting on the state space 
$\mathbb D(\mathcal H)$ of all density operators $\rho$ 
and following the time-dependent Kossakowski-Lindblad master equation of the form
\begin{eqnarray}\label{eq:GA}
\dot\rho(t) = -\big(i\, \hH(t) + \hGA(t)\big)(\rho(t))\,,
\end{eqnarray}
where $\hH(t)$ denotes the adjoint action of some time-dependent Hamiltonian $H(t)$ on
$\mathbb D(\mathcal H)$, i.e.
$
\hH(t)(\rho) := \ad_{H(t)}(\rho) = [H(t),\rho]\,.
$
For Markovianity take the noise term $\hGA(t)$ in the usual Kossakowski-Lindblad form
\begin{equation}\label{eq:GKSL}
\hGA(t) := \sum_k \gamma_k(t) \hGA_{V_k} \quad\text{with}\quad
\hGA_{V_k}(\rho) := \tfrac{1}{2} (V_k^\dagger V_k \rho + \rho V_k^\dagger V_k) - V_k\rho V_k^\dagger\,.
\end{equation}
The time dependence of $H(t)$ is brought about by adding to the (usually inevitable,
possibly unbounded) system Hamiltonian $H_0$, bounded control Hamiltonians of the type $u_j(t)H_j$ to 
give $H(t):= H_0 + \sum_{j=1}^m u_j(t) H_j$, where the control amplitudes $u_j(t)\in\mathbb{R}$ are typically
modulated in a manner at least allowing for piecewise constant controls. 

In the finite-dimensional case, an unambiguous 
separation of the dissipative part and the coherent part results by choosing the
$V_k$ traceless---as described by Kossakowski, Gorini and Sudarshan in the 
celebrated work of~\cite{GKS76}. In infinite dimensions, this separation is a bit delicate yet not 
crucial for the sequel. More important is the restriction to compact
noise terms $V_k$. 
In our case of interest, the noise terms can be individually switched on and off (as \/`bang-bang controls\/'), so
it suffices to study a single noise term\footnote{Clearly, collectively switched noise in the sense
of $\gamma(t)=\gamma_k(t)$ for all $k$ is more subtle.} 
\begin{equation}\label{eq:def_gamma_V}
\hGA(t) \big(\rho\big) := \gamma(t)\hGA_V 
\end{equation}
with $\gamma(t)\in\{0,\gamma_*\}$ and $\gamma_*>0$---w.l.o.g.~we always assume $\gamma_*=1$. 
With these stipulations, we refer to the master equation 
\eqref{eq:GA} as {\sc gksl}-equation henceforth.
In the limiting case of $\gamma(t)=0$ for all times, the control system of Eq.~\eqref{eq:GA} turns into a 
closed Hamiltonian system referred to as $\Sigma_0$, while for switchable noise with a {\em single}
$V$-term the system will be labelled $\Sigma_V$.

Note that with the identifications $A:=i\,\hH_0$, $B_j:=i\, \hH_j$, $B_{0}:= \hGA_V$ and 
$X(t)=\rho(t)$ one formally gets a standard {\em bilinear control system} 
\cite{Sontag,Elliott09} 
\begin{equation*}
\dot X(t) = -(A + \sum\limits_{j=0}^{m} u_j(t) B_j) X(t)\quad\text{with}\quad X(0)=X_0\,,
\end{equation*}
also identifying $u_{0}(t)=\gamma(t)$. This covers a broad class of quantum control problems including 
coherent and incoherent feedback~\cite{mirrahimi2004controllability,DongPetersen2010,Haroche11,Haroche13}. 
Accessibility of such bilinear Markovian quantum systems (in finite dimensions)
was analysed i.a.~in terms of symmetries in previous work~\cite{OSID17}.

In the following, we are interested in characterising the reachable sets of $\Sigma_V$ which take the form
of a {\em semigroup orbit}
\begin{equation}\label{eq:reach_def}
\reach _{\Sigma_V}(\rho_0) := \mathcal{S}_{\Sigma_V} \cdot \rho_0 := \{ F(\rho_0) \,|\, F \in \mathcal{S}_{\Sigma_V}\}\;,
\end{equation}
where $\rho_0\in\mathbb D(\mathcal H)$ denotes an arbitrary initial density operator and $\mathcal{S}_{\Sigma_V}$ is
the semigroup generated by the one-parameter semigroups 
\begin{equation*}
\big(\mathrm{e}^{-t(i\,\ad_{H_0} + i\sum\limits_{j=1}^m u_j \ad_{H_j} + u_0 \hGA_V)}\big)_{t \in \mathbb R_+}
\quad\text{with}\; u_1,\dots,u_m \in {\mathbb R}\,, \quad u_0 \in \{0,1\}\,.
\end{equation*}
If all involved operators are bounded, $\mathrm{e}^{-t(i\,\ad_{H_0} +i \sum_{j=1}^m u_j \ad_{H_j} + u_0 \hGA_V)}$
is given by the exponential series and $\reach _{\Sigma_V}(\rho_0)$ can alternatively be defined as
the collection of all endpoints $\rho(T)$, $T \geq 0$ of trajectories of \eqref{eq:GA} for piecewise
constant controls and initial value $\rho(0) = \rho_0$. 
For the general case ($H_0$ 
unbounded), defining $\reach _{\Sigma_V}(\rho_0)$ via trajectories is problematic since 
Eq.~\eqref{eq:GA} allows classical solutions only on a dense domain of initial states. Yet,
Eq.~\eqref{eq:reach_def} also works for unbounded $H_0$, as
$-(i\,\ad_{H_0} +i \sum_{j=1}^m u_j \ad_{H_j} + u_0 \hGA_V)$ does generate a unique strongly continuous semigroup
(details in Appendix \appref{D}).\vspace{4pt}

We start the discussion by $\Sigma_0$, assuming for the moment that the noise is switched off, 
i.e.~$\gamma(t)=0$. In finite dimensions such a system is {\em fully unitarily controllable} if 
it satisfies the Lie-algebra rank condition~\cite{SJ72,JS72,Bro72,Bro73,DiHeGAMM08}
\begin{equation}\label{eqn:LARC}
\langle iH_0, iH_j \,|\,j=1,2,\dots,m\rangle_{\rm Lie}=\su(\mathcal H) \quad (\text{or} \; = \uu(\mathcal H) )\,.
\end{equation}
Then reachable sets are unitary group orbits of the respective
initial states
\begin{equation*}
\reach _{\Sigma_0}(\rho_0) = \{ U\rho_0\, U^\dagger\,|\, U\in \mathcal U(\mathcal H) \}\,.
\end{equation*}
If the Lie closure $\fk := \langle iH_0, iH_j \,|\,j=1,2,\dots,m\rangle_{\rm Lie}$ in Eq.~\eqref{eqn:LARC}
is but a proper compact subalgebra $\fk \subsetneq \su(\mathcal H)$, one likewise gets a subgroup orbit 
now by limiting $U$ to elements of $\bK := \exp\fk\subsetneq \mathcal U(\mathcal H)$, see, e.g., 
\cite{DiHeGAMM08,OSID17}.

Yet already in open finite dimensional quantum systems $\Sigma_V$, it is more intricate to characterise
reachable sets: In the unital case, i.e.~for $\hGA(t)(\mathbf{1}) = 0$, 
one finds by the seminal work of \cite{Uhlm71,AlbUhlm82} and \cite{Ando89} on majorization the inclusion
\begin{equation*}
\overline{\reach_{\Sigma}(\rho_0)} \subseteq \{\rho\in\mathbb D(\mathcal H)\,|\,\rho\,{\prec}\,\rho_0\}
\end{equation*}
as used in \cite{Yuan10}. Henceforth, $\overline{\reach_{\Sigma}(\rho_0)}$ denotes 
the closure\footnote{In both, finite and infinite dimensions, there is a canonical choice for the topology 
on $\mathbb D(\mathcal H)$---we will come back to this point later.} of $\reach_{\Sigma}(\rho_0)$.
In the special case $\hGA(t) = \gamma(t)\hGA_V \neq 0$ (where unitality of $\hGA(t)$ boils
down to normality of $V$) one can obtain equality even for {\em unswitchable} noise if there
are no bounds on the coherent controls $u_k(t)$ and already the control Hamiltonians (without the drift
$iH_0$) satisfy
$
\langle  iH_j \,|\,j=1,2,\dots,m\rangle_{\rm Lie} =\su(\mathcal H),
$
a scenario we called {\em Hamiltonian controllable} (fully\/ $H$-controllable) \cite{DHKS08,ODS11}.

For many experiments this is hopelessly idealising unless one can switch off the noise---a scenario
studied below---because then one is allowed to ``use'' also the drift Hamiltonian $H_0$ for controlling 
the system in the course of noise-free evolution. However, for all physical scenarios (requiring the 
drift Hamiltonian $H_0$ for full controllability of its Hamiltonian part) with sizeable constant noise, 
the above inclusion is far from being tight and---even worse---the overestimation of the reachable set 
increases with system size.
In these cases, {\em Lie-semigroup} techniques help to estimate the
reachable set \cite{DHKS08,ODS11}.

Yet there are indeed instances of unitarily controllable systems of the type $\Sigma_V$ in which the 
noise can be switched as bang-bang control. An important experimental incarnation are 
superconducting qubits coupled to an open transmission line~\cite{Mart14}. Then, for normal $V$, one
can saturate the above inclusion to get 
$\overline{\reach(\rho_0)} = \{\rho\in\mathbb D(\mathcal H)\,|\,\rho\,{\prec}\,\rho_0\}$ 
as shown in~\cite{BSH16,OSID17}. In another extreme, $\hGA_V$ models coupling the system to a bath of temperature zero 
(entailing $V$ is the nilpotent matrix $\sigma_-$). In this case 
$\overline{\reach(\rho_0)} = \mathbb D(\mathbb C^n)$ \cite{CDC19}.

\bigskip
Here the goal is to transfer the former result (with normal $V$) from finite to infinite-dimensional   
systems on separable complex Hilbert spaces $\mathcal H$.

\subsection{Main Result}\label{sec:overview_main}

\noindent
In infinite dimensions, establishing unitary controllability for $\Sigma_0$ is 
more intricate. One of the most general results currently known 
is the following~\cite{keyl18InfLie}:

\bigskip
\noindent
Let $H_0, . . . ,H_m$ be selfadjoint operators on a separable Hilbert space $\mathcal H$.
Further assume that
\begin{enumerate}
\item[(1)] $H_0$ is bounded or unbounded, but has only pure point spectrum. The eigenvalues
$x_k, k \in \mathbb N$ are non-degenerate and rationally independent.
\item[(2)] The operators $H_1, . . . ,H_m$ are bounded and the set $\{H_1, . . . ,H_m\}$ 
is connected\footnote{This means that the associated graph (which roughly
speaking indicates where a transition from energy level $k$ to $l$ is possible) has to be
connected, cf.~\cite{keyl18InfLie}.} with respect to the complete set of eigenvectors $\phi_k \in \mathcal H$,
$k \in \mathbb N$ of $H_0$.
\end{enumerate}
Then the unitary system
\begin{equation}\label{eq:unitary}
\dot U(t) = - iH(t)U(t)\quad\text{with}\quad U(0) = \mathbf{1}\,,
\end{equation}
is {\em strongly approximately operator controllable} in the following sense:
%
\begin{defi}\label{def:strong_cont}
The unitary control system (\ref{eq:unitary}) is called {\em strongly approximately operator controllable}, 
if the strong closure (in $\mathcal U(\mathcal H)$) of the reachable set $\reach(\mathbf{1})$ co{\"i}ncides with $\mathcal U(\mathcal H)$. 
\end{defi}


The result can be generalized to eigenvalues $x_k$, $k \in \mathbb{N}$ with finite multiplicities, but 
this requires more technical conditions on the control Hamiltonians: We have to ensure that trace-free 
finite rank operators commuting with all eigenprojections of $H_0$ are contained in the strong closure of
the Lie algebra generated by the $H_j$, $j=1,\dots,m$. More challenging are drift
Hamiltonians with rationally dependent eigenvalues. However, they can be studied in terms of certain
non-Abelian von Neumann algebras; cf.~\cite{keyl18InfLie} for details. Similar results were derived
earlier in terms of Galerkin approximations 
in \cite{boscain2012weak} and were refined more recently in \cite{caponigro2018exact}.

If all Hamiltonians (including $H_0$) are bounded, we use
approximate versions of the Lie algebra rank condition, the most
straightforward one being
\begin{equation}\label{eq:strong_cont}
 \clos{\langle iH_0, iH_j, i\mathbf{1}\,|\,j=1,2,\dots,m\rangle_{\rm Lie}} =  \uu(\mathcal H)\,.
\end{equation}
Using the continuity of the exponential map in the strong topology
\cite{keyl18InfLie} it is easy to see that this condition is
sufficient for strong operator controllability of
(\ref{eq:unitary}). Our conjecture is that it is not necessary, but
counter examples are not known (their construction is subject of
current research). Stronger types of convergence can be achieved if
all the Hamiltonians $H_j$, $j=0,\dots,m$ are even compact. Since the
strong closure of the algebra $\mathcal{K}(\mathcal{H})$ of compact
operators is $\mathcal{B}(\mathcal{H})$, it is clear that
Eq. (\ref{eq:strong_cont}) is implied by
\begin{equation*}
  \clou{\langle iH_0, iH_j\,|\,j=1,2,\dots,m\rangle_{\rm Lie}} = \uu\bigl(\mathcal{K}(\mathcal{H})\bigr)\,,
\end{equation*}
where the closure is taken in the uniform (operator norm) topology. If the
other implication also holds is still unclear, but 
unlikely. Note that a compact operator can be the strong limit of
a sequence in $\mathcal{K}(\mathcal{H})$ without being the uniform
limit. 

Let us fix some final notations with regard to the {\sc gksl}-equation: $\mathcal B(\mathcal H)$ and $\mathcal B^1(\mathcal H)$
denote the spaces of all bounded and trace-class operators on $\mathcal H$, respectively. Thus 
$\mathbb D(\mathcal H) \subset \mathcal B^1(\mathcal H)$ is precisely the set of all positive semi-definite 
(selfadjoint) trace-class operators with trace~$1$. Moreover, $\|\cdot\|_1$ stands for the trace norm
on $\mathcal B^1(\mathcal H)$ (see Appendix \appref{A} for more detail on the trace class). 
To begin with, all Hamiltonians $H_0, H_j$ are assumed to be taken from $\mathcal B(\mathcal H)$, while
later $H_0$ may be any unbounded selfadjoint operator. 

In this setting, the operator solutions of \eqref{eq:master-map} are globally well-defined (with respect to 
$t \in \R$) for arbitrary piecewise continuous controls (even more irregular controls are admissible)
and for each fixed $t\in\mathbb R^+$ the corresponding map is ultraweakly continuous (cf.~footnote \ref{footnote_uw}) and {\sc cptp}.  
In particular for constant controls they form uniformly continuous semigroups of ultraweakly 
continuous {\sc cptp}-maps, \cite[Thm.~1 \& 2]{Lind76}.

\bigskip
\noindent
With these notions and notations 
and taking majorization from finite to infinite dimensions by way of
sequence spaces as introduced by Gohberg and Markus~\cite{gohberg66} (see Sec.~\ref{sec:major_fin_infin}), 
our main result for $\hGA(t) =\gamma(t) \hGA_V(t)$ reads:

\begin{thm}\label{thm_normal_V}
Given the Markovian control system $\Sigma_V$ 
\begin{equation*}
\dot\rho(t) = - i \Big[H_0 + \sum_{j=1}^mu_j(t)H_j,\rho\Big] - 
\gamma(t) \big(\tfrac{1}{2}(V^\dagger V \rho + \rho V^\dagger V) - V\rho V^\dagger \big)\,,\;\text{where}
\end{equation*}
\begin{itemize}
\item[(1)] the drift $H_0$ is selfadjoint and the controls $H_1,\ldots,H_m$ are selfadjoint and bounded,
\item[(2)] the Hamiltonian part $\Sigma_0$ is strongly (approximately) operator controllable 
in the sense of Def.~\ref{def:strong_cont},
\item[(3)] the noise term $V\in\mathcal K(\mathcal H)\setminus\lbrace 0\rbrace$ is compact, normal and switchable by $\gamma(t) \in\{0,1\}$.
\end{itemize}
Then the $||\cdot||_1$-closure of the reachable set of any initial state \mbox{$\rho(0)=\rho_0\in\mathbb D(\mathcal H)$} under 
the system $\Sigma_V$ 
exhausts all states majorized by the initial state $\rho_0$ 
\begin{align*}
\clon{\reach_{\Sigma_V}(\rho_0)}=\lbrace \rho\in\mathbb D(\mathcal H)\,|\,\rho\prec\rho_0\rbrace\,.
\end{align*}
\end{thm}
%

\bigskip
\bigskip
\noindent
In order to arrive at this result, the paper is organised as follows:
Section~\ref{sec:maj_setconv} first takes majorization from finite to infinite dimensions in~\ref{sec:major_fin_infin}
before combining ideas of von Neumann with \mbox{$C$-numerical} ranges for majorization in infinite dimensions~\ref{sec:set_conv}
Section~\ref{sec:main_res} then presents the idea of the main theorem, the proof details themselves being relegated
to the Appendix. Appendix \appref{A} contains technical basics, while Appendix \appref{B} gives the proofs to Section~\ref{sec:set_conv}
Finally Appendix \appref{C} provides the proof of the main theorem 
for bounded $H_0$, while Appendix \appref{D} relaxes it to unbounded $H_0$.

\section{From Majorization via $C$-Numerical Range to Reachability}\label{sec:maj_setconv}
\subsection{Majorization in Finite and Infinite Dimensions}\label{sec:major_fin_infin}

Generalizing majorization to infinite dimensions is somewhat delicate. Following \cite{gohberg66}, one may 
define majorization first on the space of all real null sequences~$c_0(\mathbb N)$ and then on the space of all 
absolutely summable sequences~$\ell^1(\mathbb N)$. As we 
need a concept of majorization on density 
operators, for our purposes it suffices to introduce majorization solely on 
the summable sequences of non-negative numbers~$\ell^1_+(\mathbb N)$, 
which is rather intuitive.

In the notation of \cite{Ando89,MarshallOlkin},
take a real vector $x=(x_1,\ldots,x_n)^T\in\mathbb R^n$ and let $x^\downarrow=( x_1^\downarrow,\ldots,x_n^\downarrow)^T$ 
denote its decreasing re-arrangement
$
x_1^\downarrow\geq x_2^\downarrow\geq \ldots\geq x_n^\downarrow\,.
$
For two vectors $x,y\in\mathbb R^n$, we say $x$ is \textit{majorized} by $y$ (written $x\prec y$) if $\sum_{j=1}^k x_j^\downarrow\leq\sum_{j=1}^k y_j^\downarrow$ for all $k=1,\ldots,n-1$ and $\sum_{j=1}^n x_j=\sum_{j=1}^n y_j$. By definition $x\prec y$ 
depends only on the entries of $x$ and $y$ but not on their
initial arrangement, so $\prec$ is permutation invariant. 

Now for sequences $x\in\ell^1_+(\mathbb N)$, this re-arrangement procedure works just the same way,
and all the non-zero entries of $x$ are again contained within the rearranged sequence $x^\downarrow$. However, be aware that $x$ and $x^\downarrow$ may differ in the number of their zero entries. 

\begin{defi}\label{defi_maj}
Consider $x,y\in\ell^1_+(\mathbb N)$ and $\rho,\omega\in\mathbb D(\mathcal H)$.
\begin{itemize}
\item[(a)] We say that $x$ is majorized by $y$, denoted by $x\prec y$, if the sum inequalities $\sum_{j=1}^k x_j^\downarrow \leq\sum_{j=1}^k y_j^\downarrow$ 
hold for all $k\in\mathbb N$, and $\sum_{j=1}^\infty x_j=\sum_{j=1}^\infty y_j$.\vspace{4pt}
\item[(b)] $\omega$ majorizes $\rho$, denoted by $\rho\prec\omega$, 
if $\lambda^\downarrow(\rho)\prec\lambda^\downarrow(\omega)$ where $\lambda^\downarrow(\cdot)\in\ell^1_+(\mathbb N)$ denotes the (non-modified) 
eigenvalue sequence\footnote{Usually, the 
eigenvalue sequence $\lambda^\downarrow(T)$ of a compact operator $T$ on $\mathcal H$ is obtained by 
arranging its non-zero eigenvalues in the decreasing order of their magnitudes and each eigenvalue 
is repeated as many times as its (necessarily finite) algebraic multiplicity. If the spectrum of $T$ is finite 
itself, then the sequence is filled with zeros, cf. \cite[Ch.~15]{MeiseVogt97en}. However, in order to get 
the result of Lemma \ref{lemma_berb_4_6} with respect to an orthonormal basis (and not just an 
orthonormal system), and also to properly define the $C$-spectrum of $T$ later on, a modified 
eigenvalue sequence has to be introduced, as in \cite[Ch.~3.2]{dirr_ve}.

If the range of $T$ is infinite-dimensional and the kernel of $T$ finite-dimensional
then put $\operatorname{dim}(\operatorname{ker}T)$ zeros at the beginning of the eigenvalue
sequence of $T$. If the range and the kernel of $T$ are infinite-dimensional, mix infinitely many zeros into the
eigenvalue sequence of $T$ (since for the $C$-spectrum arbitrary permutations
will be applied to the modified eigenvalue sequence, we need not specify this mixing 
procedure further). If the range of $T$ is finite-dimensional, leave the eigenvalue sequence of $T$ unchanged.\label{footnote_1}} of the respective state.
\end{itemize}
\end{defi}
\begin{remark}\label{rem_maj}
In Definition \ref{defi_maj} (b) it does not matter whether one considers the usual (non-modified) or the modified 
eigenvalue sequence (for the purpose of this remark denoted by $\lambda^\downarrow$ and $\lambda_m$, respectively). 
More precisely, these sequences by construction share the same non-zero entries so $\lambda^\downarrow={\lambda^\downarrow_m},$.

\end{remark}

As in finite dimensions, majorization in infinite dimensions has a number of
different characterizations, the following two being particularly advantageous
for our purposes.
\begin{lemma}[\cite{Li13}, Thm.~3.3]\label{lemma_li}
For $\rho,\omega\in\mathbb D(\mathcal H)$ the following 
are equivalent:
\begin{itemize}
\item[(a)] $\rho\prec\omega$.
\item[(b)] There exists a bi-stochastic quantum map $T\in \mathbb S(\mathcal H)$ (cf.~Def.~\ref{def1} in Appendix \appref{A}) such that $T(\omega)=\rho$. 
\end{itemize}
\end{lemma}
\begin{proposition}\label{prop_schur_horn_gohberg}
Let $x,y\in\ell^1_+(\mathbb N)$ be non-increasing and let $(e_n)_{n\in\mathbb N}$ be some 
orthonormal basis of $\mathcal H$. Then the following statements are equivalent:
\begin{itemize}
\item[(a)] $x\prec y$
\item[(b)] There exists a selfadjoint $A\in\mathcal B^1(\mathcal H)$ with diagonal entries $(x_n)_{n\in\mathbb N}$ 
and eigenvalues $(y_n)_{n\in\mathbb N}$.
\item[(c)] There exists a unitary $U\in\mathcal B(\mathcal H)$ such that $U\operatorname{diag}(y)U^\dagger $ has 
diagonal entries $(x_n)_{n\in\mathbb N}$.
\end{itemize}
Here, ``diagonal'' always refers to the orthonormal basis $(e_n)_{n\in\mathbb N}$, i.e.~the map $\operatorname{diag}:\ell^1_+(\mathbb N)\to \mathcal B^1(\mathcal H)$ is given by $x\mapsto \sum_{n=1}^\infty x_n\langle e_n,\cdot\rangle e_n$.
\end{proposition}
\begin{proof}
``(a) $\Rightarrow$ (b)'': Assume $x\prec y$. By \cite[Prop.~IV]{gohberg66} there exist 
orthonormal bases $(\phi_n)_{n\in\mathbb N}$ and $(\psi_n)_{n\in\mathbb N}$ of $\mathcal H$ such that 
$H=\sum_{n=1}^\infty y_n\langle\psi_n,\cdot\rangle\psi_n$ satisfies $\langle \phi_n,H\phi_n\rangle=x_n$ 
for all $n\in\mathbb N$. Consider the unitary operator $U\in\mathcal B(\mathcal H)$ 
which transforms $(\phi_n)_{n\in\mathbb N}$ into $(e_n)_{n\in\mathbb N}$, 
then $A:=UHU^\dagger \in\mathcal B^1(\mathcal H)$ does the job. ``(b) $\Rightarrow$ (a)'': 
follows from \cite{Fan49}. ``(b) $\Leftrightarrow$ (c)'': The statement is obvious.
\end{proof}

We conclude with a classical result on sub-majorization (without proof) which will 
be needed in the following subsection.

\begin{lemma}[\cite{MarshallOlkin}, 3.H.3.b]\label{lemma_wprec}
Let $x,y\in\mathbb R^n$ such that $\sum_{j=1}^k x_j^\downarrow\leq\sum_{j=1}^k y_j^\downarrow$ for all $k=1,\ldots,n$. 
Then for arbitrary $c_1\geq c_2\geq\ldots\geq c_n\geq 0$ one has
$
\sum\nolimits_{j=1}^n c_jx_j^\downarrow \leq \sum\nolimits_{j=1}^n c_jy_j^\downarrow \,.
$
\end{lemma}

\subsection{Combining a von Neumann Idea with $C$-Numerical Ranges}\label{sec:set_conv}


In finite dimensions, Ando \cite[Thm.~7.4]{Ando89} has shown that majorization can
be characterized in an elegant way via the $C$-numerical range~\cite{Li94,TSING-96}
$$
W_C (T):=\lbrace \operatorname{tr}(CU^\dagger TU)\,|\,U\in\mathcal U(\mathcal H)\rbrace
$$
with $C\in\mathcal B^1(\mathcal H)$, $T\in\mathcal B(\mathcal H)$ and $\mathcal U(\mathcal H)\subset\mathcal B(\mathcal H)$ being the unitary group on $\mathcal H$. Here, we generalize
his approach to infinite dimensions (Prop.~\ref{thm_1} below) using a recent result in \cite{dirr_ve}.
Later on, this characterization will greatly simplify handling continuity properties of  
majorization (cf.~Lemma \ref{lemma_maj_closed}).

For our purpose, we need a relation connecting the $C$-numerical range and
$C$-spectrum of a compact operator $T$ given by
\begin{align*}
P_C(T):=\Big\lbrace \sum\nolimits_{j=1}^\infty \lambda_j(C)\lambda_{\sigma(j)}(T) \,\Big|\, \sigma:\mathbb N \to\mathbb N \text{ is a permutation}\Big\rbrace
\end{align*}
on one hand-side with $(\lambda_j(C))_{j\in\mathbb N}$ and $(\lambda_j(T))_{j\in\mathbb N}$ being the 
modified eigenvalue sequences (cf. footnote \ref{footnote_1}) of $C$ and $T$ on the other.
Note that each element in $W_C(T)$ and $P_C(T)$ is bounded by $\|C\|_1\|T\|_{\text{op}}$---thus the closures of $W_C(T)$ and $P_C(T)$ constitute compact subsets of $\mathbb C$.
%
%
%
%

If $C,T$ are normal, one has the inclusion
$\overline{W_C(T)} \subseteq \operatorname{conv}(\overline{P_C(T)})$. Yet under further assumptions on 
the operators 
one can even achieve equality.

\begin{lemma}[\cite{dirr_ve}, Coro.~3.1]\label{lemma_4}
Let $C\in\mathcal B^1(\mathcal H)$ and $T\in\mathcal K(\mathcal H)$ both be normal, 
such that the eigenvalues of $C$ are collinear, i.e.~the eigenvalues all lie on a common line. 
Then
$
\overline{W_C(T)}=\operatorname{conv}(\overline{P_C(T)})\,.
$
\end{lemma}

\noindent
In fact, Lemma \ref{lemma_4} induces a von Neumann-type of trace (in-)equality \cite{NEUM-37}
for compact, selfadjoint operators. Its proof is in Appendix \appref{B}.

\begin{corollary}\label{cor_von_Neumann}
Let $C\in\mathcal B^1(\mathcal H)$ and $T\in\mathcal K(\mathcal H)$ both be selfadjoint.
Then 
\begin{equation*}
\sup_{U\in\mathcal U(\mathcal H)}\operatorname{tr}(C U^\dagger TU)
= \sum\nolimits_{j=1}^\infty \lambda_j^\downarrow(C^+) \lambda_{j}^\downarrow(T^+) 
+ \sum\nolimits_{j=1}^\infty \lambda_j^\downarrow(C^-) \lambda_{j}^\downarrow(T^-)\,,
\end{equation*}
where $\lambda_j^\downarrow(C^+)$, $\lambda_{j}^\downarrow(T^+)$ and $\lambda_j^\downarrow(C^-)$, $\lambda_{j}^\downarrow(T^-)$
denote the decreasing eigenvalue sequences of the positive semi-definite operators $C^+$, $T^+$ and $C^-$, $T^-$, respectively, where $C=C^+-C^-$ and $T=T^+-T^-$ as usual.
\end{corollary}

To simplify notation, we use the following abbreviation.

\begin{defi}\label{def_KC}
Let $C\in\mathcal B^1(\mathcal H)$, $T\in\mathcal B(\mathcal H)$ both be selfadjoint. We define
$
K_C(T):=\sup_{U \in\mathcal U(\mathcal H)}\operatorname{tr}(C U^\dagger TU)\in\mathbb R
$
or, equivalently, $K_C(T) :=\sup W_C(T) = \max\overline{W_C(T)}$.
\end{defi}

\noindent
Note that if $C$ and $T$ are positive semi-definite, then $K_C(T)$ turns into 
the \mbox{$C$-numerical} radius $r_C(T)$ of $T$. Now this definition gives rise to the following
result, whose finite-dimensional analogue can be found, e.g., in \cite[Thm.~7.4]{Ando89}.
\begin{proposition}\label{thm_1}
For $\rho,\omega\in\mathbb D(\mathcal H)$ the following statements are equivalent.
\begin{itemize}
\item[(a)] $\rho \prec \omega$\vspace{2pt}
\item[(b)] $K_\rho(T)\leq K_\omega(T)$ for all selfadjoint $T\in\mathcal K(\mathcal H)$.
\end{itemize}
\end{proposition}
\begin{proof}

``(a) $\Rightarrow$ (b)'': Keeping in mind that $\rho,\omega\geq 0$, Coro.~\ref{cor_von_Neumann} yields
\begin{align*}
K_\rho(T) = \max \overline{W_\rho(T)} = \sum\nolimits_{j=1}^\infty \lambda_j^\downarrow(\rho) \lambda_{j}^\downarrow(T^+)
\end{align*}
and similarly for $K_\omega(T)$. 
Moreover, Lemma \ref{lemma_wprec} implies
\begin{equation*}
\sum\nolimits_{j=1}^n \lambda_j^\downarrow(\rho) \lambda_{j}^\downarrow(T^+)
\leq \sum\nolimits_{j=1}^n \lambda_j^\downarrow(\omega) \lambda_{j}^\downarrow(T^+)
\end{equation*}
for all $n \in \mathbb N$ and thus it follows  $K_\rho(T)\leq K_\omega(T)$
for all selfadjoint $T\in\mathcal K(\mathcal H)$.

\medskip
\noindent
``(b) $\Rightarrow$ (a)'':  Let $k\in\mathbb N$ and let $(e_n)_{n\in\mathbb N}$  be any orthonormal basis of $\mathcal H$. Consider the (finite-rank) projection
$\Pi_k=\sum\nolimits_{j=1}^k\langle e_j,\cdot\rangle e_j$. As $\Pi_k$ is compact and selfadjoint with 
eigenvalues $1$ (of multiplicity $k$) and $0$ (of infinite multiplicity), Coro.~\ref{cor_von_Neumann} 
yields
$
K_\rho(\Pi_k)=\sum\nolimits_{j=1}^k\lambda_j^\downarrow(\rho) 
$
and
$
K_\omega(\Pi_k)=\sum\nolimits_{j=1}^k\lambda_j^\downarrow(\omega)\,.
$
Now by assumption, one has
\begin{align*}
\sum\nolimits_{j=1}^k\lambda_j^\downarrow(\rho)=K_\rho(\Pi_k)\leq K_\omega(\Pi_k)=\sum\nolimits_{j=1}^k\lambda_j^\downarrow(\omega)
\end{align*}
for all $k\in\mathbb N$ which shows $\rho\prec\omega$ and thus concludes this proof.
\end{proof}

\section{Idea behind the Main Result}\label{sec:main_res}

Below we sketch the proof of our main result Thm.~\ref{thm_normal_V}. A full proof will 
be given in Appendix \appref{C}. Here, we sketch central ideas and key 
lemmas, the proofs of which are either straightforward 
or postponed to Appendices
\appref{C} and \appref{D}. For convenience, let us first recall the precise statement of Thm.~\ref{thm_normal_V}.

\thmskip
\noindent
\textbf{Theorem \ref{thm_normal_V}.} \textit{Given the Markovian control system $\Sigma_V$ 
\begin{equation*}
\dot\rho(t) = - i \Big[H_0 + \sum_{j=1}^mu_j(t)H_j,\rho\Big] - 
\gamma(t) \big(\tfrac{1}{2}(V^\dagger V \rho + \rho V^\dagger V) - V\rho V^\dagger \big)\,,\;\text{where}
\end{equation*}
\begin{itemize}
\item[(1)] the drift $H_0$ is selfadjoint and the controls $H_1,\ldots,H_m$ are selfadjoint and bounded,
\item[(2)] the Hamiltonian part $\Sigma_0$ is strongly (approximately) operator controllable 
in the sense of Def.~\ref{def:strong_cont},
\item[(3)] the noise term $V\in\mathcal K(\mathcal H)\setminus\lbrace 0\rbrace$ is compact, normal and switchable by $\gamma(t)\in\{0,1\}$.
\end{itemize}
Then the $||\cdot||_1$-closure of the reachable set of any initial state \mbox{$\rho(0)=\rho_0\in\mathbb D(\mathcal H)$} under 
the system $\Sigma_V$ 
exhausts all states majorized by the initial state $\rho_0$ 
\begin{align*}
\clon{\reach_{\Sigma_V}(\rho_0)}=\lbrace \rho\in\mathbb D(\mathcal H)\,|\,\rho\prec\rho_0\rbrace\,.
\end{align*}}%
\noindent%
The following lemmas play a crucial role in the proof of Theorem~\ref{thm_normal_V}.
The first one reveals a beautiful eigenspace structure of the noise generators $\hGA_V$ whenever
$V$ is normal and compact, and it follows by direct computation.
\begin{lemma}\label{lemma_normal_generator}
Let $V\in\mathcal K(\mathcal H)$ be normal, $(f_j)_{j\in\mathbb N}$ its orthonormal eigenbasis
and $(v_j)_{j\in\mathbb N}$ its modified eigenvalue sequence, hence 
$ V =\sum_{j=1}^\infty v_j\langle f_j,\cdot\rangle f_j$ (cf.~Appendix \appref{A}). 
Then for all $X\in\mathcal B(\mathcal H)$, the noise operator $\hGA_V$  given by Eq.~\eqref{eq:def_gamma_V} 
acts like
\begin{equation}\label{eq:action_noise}
\begin{split}
\langle f_j,\hGA_V(X) f_k\rangle 
& =\big(\tfrac{1}{2}{|v_j-v_k|^2}-i\operatorname{Im}(v_j\overline{v_k})\big)  \langle f_j,Xf_k\rangle
\end{split}
\end{equation}
for all $j,k\in\mathbb N$. In particular, each rank-$1$ operator of the form  
$\langle f_k,\cdot\rangle f_j$ is an eigenvector of $\hGA_V$ to the eigenvalue 
$\tfrac{1}{2}{|v_j-v_k|^2}-i\operatorname{Im}(v_j\overline{v_k})$ and the kernel of
$\hGA_V$ contains $ \operatorname{span}\{\langle f_j,\cdot\rangle f_j \,|\, j \in\mathbb N\}$.
Moreover, it follows
\begin{equation*}
\exp(-t\hGA_V)(\langle f_k,\cdot\rangle f_j)
= \exp\big(-\tfrac{t}2 |v_j-v_k|^2\big)\exp(it\operatorname{Im}(v_j\overline{v_k}))\langle f_k,\cdot\rangle f_j
\end{equation*}
for all $t\in\mathbb R$ and $j,k\in\mathbb N$. 
\end{lemma} 


The following lemmas provide two crucial approximation results.

\begin{lemma}\label{lemma_maj_closed}
For all $\rho_0\in\mathbb D(\mathcal H)$ the set 
$\lbrace \rho\in\mathbb D(\mathcal H)\,|\,\rho\prec\rho_0\rbrace$ 
is closed w.r.t.~the trace norm $\|\cdot\|_1$.
\end{lemma}

\begin{lemma}[Unitary channel approximation]\label{lemma_unitary_channel_approx}
Consider a subset $R\subseteq \mathcal U(\mathcal H)$ of the unitary group of $\mathcal H$ such 
that $\overline{R}^s=\mathcal U(\mathcal H)$, i.e.~its strong closure relative to $\mathcal U(\mathcal H)$ yields the full group. 
Furthermore, let $\rho\in\mathbb D(\mathcal H)$ and $U\in\mathcal U(\mathcal H)$.\\
Then for all $\varepsilon>0$ one can find $\tilde U\in R$ such that
$
\|U\rho \,U^\dagger-\tilde U\rho\,\tilde U^\dagger\|_1<\varepsilon\,.
$
\end{lemma}

Now in our control setting we do not have direct access to the ``pure'' noise
generator $-\hGA_V$. However, 
we may use the Lie-Trotter product formula (cf.~\cite[Thm.~VIII.29]{ReedSimonI}) 
to approximate the noise dynamics $(\exp(-t\hGA_V))_{t\in\mathbb R_+}$:

\begin{lemma}[Trotter trick]\label{rem_trotter}
For $u_1(t)=\ldots=u_m(t)=0$ and $\gamma(t)=1$
(i.e.~noise only) the operator solution of \eqref{eq:GA} reads 
$(\exp(-it\hH_0-t\hGA_V))_{t\in\mathbb R_0^+}$. 
Then for $t\geq 0$ (uniformly on bounded intervals) one has
\begin{equation*}
\lim_{n\to\infty}\Big\Vert\Big( \exp\Big(\frac{it\hH_0}{n}\Big)\exp\Big(\frac{-it\hH_0-t\hGA_V}{n}\Big) \Big)^n-\exp(-t\hGA_V)\Big\Vert_\text{op}=0\,.
\end{equation*}
\end{lemma}
\noindent Thus given a time $t\geq 0$ and precision $\varepsilon>0$, 
to ``simulate'' $\exp(-t\hGA_V)$ within this precision it suffices 
to apply the noisy evolution as well as the unitary channel 
$ \exp(it\hH_0/N)$ to the system---in an alternating manner, $N$ times (for sufficiently large $N\in\mathbb N$). ---
Now 
we are ready to outline the proof of Thm.~\ref{thm_normal_V}.

\begin{proof}[Sketch of the proof of Theorem \ref{thm_normal_V}]
``$\subseteq$'': As $V$ is assumed to be normal one has $(i\hH(t) + \gamma(t)\hGA_V)(\mathbf{1})=0$
at all times so the operator solution of Eq.~\eqref{eq:GA} is in $\mathbb S(\mathcal H)$, i.e.~a 
bi-stochastic quantum map and one can never leave the set of states majorized by $\rho_0$ 
(cf.~Lemma \ref{lemma_li}). By Lemma \ref{lemma_maj_closed}, the $\|\cdot\|_1$-closure yields
$$
\clon{\reach_{\Sigma_V}(\rho_0)}
	\subseteq 
	\clon{\lbrace \rho\in\mathbb D(\mathcal H)\,|\,\rho\prec\rho_0\rbrace}
    =\lbrace \rho\in\mathbb D(\mathcal H)\,|\,\rho\prec\rho_0\rbrace\,.
$$

\noindent
``$\supseteq$'': As $V\in\mathcal K(\mathcal H)$ is normal we can diagonalize it (see Appendix \appref{A})
with orthonormal eigenbasis $(e_j)_{j \in \mathbb N}$.
Now, let $\varepsilon>0$ and $\rho\in\mathbb D(\mathcal H)$ with $\rho\prec\rho_0$ be given. We have 
to find $\rho_F\in \reach_{\Sigma_V}(\rho_0)$ such that $\|\rho-\rho_F\|_1<\varepsilon$. 
By assumption there exist $x,y\in\ell^1_+(\mathbb N)$, $x,y\neq 0$ as well as $W_1,W_2\in\mathcal U(\mathcal H)$ such that
$
\rho=W_1\operatorname{diag}(x)W_1^\dagger$, $ \rho_0=W_2\operatorname{diag}(y)W_2^\dagger
$ with $x\prec y$. Here, $\operatorname{diag}$ refers to the above eigenbasis of $V$. Applying Prop.~\ref{prop_schur_horn_gohberg} to $x,y$ gives us unitary $U\in\mathcal B(\mathcal H)$ such that $U\operatorname{diag}(y)U^\dagger$ has diagonal entries $(x_n)_{n\in\mathbb N}$. The proof roughly consists of three steps shown here:
\begin{equation}\label{eq:steps_idea}
\begin{split}
\rho_0=W_2\operatorname{diag}(y)W_2^\dagger\overset{\text{Step }1}\longrightarrow U \operatorname{diag}(y)U^\dagger
 &\overset{\text{Step }2}\longrightarrow \operatorname{diag}(x)\\
& \overset{\text{Step }3}\longrightarrow W_1\operatorname{diag}(x)W_1^\dagger=\rho\,.
\end{split}
\end{equation}
Step 1 and 3 merely apply a unitary channel; assuming strong operator controllability, 
we may use unitary channels giving the target state with arbitrary precision (cf.~Lemma 
\ref{lemma_unitary_channel_approx}).
Step 2 
is about getting rid of all off-diagonal elements of $U \operatorname{diag}(y)U^\dagger$ to reach 
$\operatorname{diag}(x)$ by applying pure noise $\exp(-t\hGA_V)$ in the limit $t\to\infty$ (cf.~Lemma 
\ref{lemma_normal_generator}). As expected there are a few delicate issues:
\begin{itemize}
\item 
We have no access to pure noise, as in our setting we cannot switch off $H_0$. 
Yet by a Trotter-type 
argument we can approximate the desired noise with arbitrary precision, cf.~Lemma \ref{rem_trotter}
and Lemma \ref{lemma_trotter_approx}.
\item 
If the eigenvalues of $V$ are not pairwise different, there are some ``matrix'' elements left 
untouched by the noise as a consequence of \eqref{eq:action_noise}. So one may need permutation channels 
(which in particular are unitary) to rearrange those elements into ``spots'' where the noise affects them. 
\item 
As in Step 1 and 3 we have to approximate these permutation channels. Here we use the 
approximation property of the trace class (cf.~Lemma~\ref{lemma_block_approx}), i.e.~we invoke 
decoherence on a sufficiently large but finite ``block'' of the density operator so we only need
finitely many permutations.
\item
Applying Prop.~1 requires that $\rho_0$, $\rho$ are unitarily diagonalized so that the 
original and the modified eigenvalue sequences of these states co{\"i}ncide (which either means the states 
are finite-rank or have trivial kernel)---else the zeros that have to be added for the modified eigenvalue sequence 
prevent this. In the latter case we can proceed to states $\rho'$, $\rho_0'$ which satisfy the assumptions of 
Prop.~1 and which are close (in trace norm) to the original states, and execute the scheme of Eq.~\eqref{eq:steps_idea}.
\end{itemize}
Altogether this is enough to perform the scheme suggested in Eq.~\eqref{eq:steps_idea} with arbitrary 
precision. So $\rho$ $\prec\rho_0$ is in the $\|\cdot\|_1$-closure of the reachable set.
The full proof with all detail is in Appendices \appref{C} and \appref D.
\end{proof}

\section{Conclusions and Outlook}
For the first time, here we have derived sufficient conditions under which a quantum dynamical 
system can actually reach (in the closure) all quantum states majorized by the respective initial 
state in an infinite dimensional quantum system following a 
controlled Markovian master equation. 
To this end, we have extended the standard unital {\sc gksl} master equation to an infinite dimensional 
bilinear control system $\Sigma_V$ the unitary part of which has to be operator controllable
and the dissipative part (generated by a single normal compact noise term $V$) 
has to be bang-bang switchable.
This takes recent results on finite dimensional
systems~\cite{BSH16,OSID17} to infinite dimensions. --- While the generalization from a single such $V$ to 
several commuting compact noise terms $V_k$ is obvious,
a generalization beyond compact $V$ seems challenging.
One may also relax considerations to 
weak-$*$ continuity of the semigroup,
which goes beyond the standard {\sc gksl}-equation, as pursued by Carbone, Fagnola \cite{CarboneFagnola01} 
and more recently by Siemon, Holevo and Werner~\cite{OSID_Werner_17}. 


For \textit{applying} the results to broader classes of physical systems, one may think of further 
generalizations.
The current setup restricts us to (possibly unbounded) system Hamiltonians $H_0$ 
with discrete spectrum such as
bound systems where particles are trapped within an unbounded potential (e.g., harmonic oscillators). 
To look at more interesting setups where processes like ionization,
tunneling and evaporation play a role, we have to use operators with
continuous spectrum. However, in this area even coherent control is
not understood well enough (if at all). Closing this gap is therefore
an obvious (yet non-trivial\hspace{.3mm}!\/) next step. 

Thus the spirit of
Sudarshan still promises insightful results to come.

\section*{Acknowledgements}
This work was supported 
by the Bavarian excellence network {\sc enb}
via the \mbox{International} PhD Programme of Excellence
{\em Exploring Quantum Matter} ({\sc exqm})
and by {\em Deutsche Forschungsgemeinschaft} ({\sc dfg}, German Research
Foundation) under Germany’s Excellence Strategy {\sc exc}-2111–390814868.

\section*{Appendix \app{A}: Notation and Basics}
For a comprehensive introduction to infinite-dimensional separable Hilbert spaces and 
Schatten-class operators we refer to, e.g.,~\cite{berberian1976,MeiseVogt97en,ReedSimonI}. 
As we will encounter compact normal operators repeatedly, let us first recap the 
well-known diagonalization result.
\begin{lemma}[\cite{berberian1976}, Thm.~VIII.4.6]\label{lemma_berb_4_6}
Let $T\in\mathcal K(\mathcal H)$ be normal, i.e.~$T^\dagger T=TT^\dagger$. Then there exists an 
orthonormal basis $(e_n)_{n\in\mathbb N}$ of $\mathcal H$ such that
\begin{align*}
T=\sum\nolimits_{j=1}^\infty \tau_j\langle e_j,\cdot\rangle e_j
\end{align*}
where $(\tau_j)_{j\in\mathbb N}$ is the modified eigenvalue sequence (cf.~footnote \ref{footnote_1}) of $T$.
\end{lemma}

Moreover, recall that the set of trace-class operators is given by
\begin{equation*}
\mathcal B^1(\mathcal H):=\lbrace C\in\mathcal B(\mathcal H)\,|\,\Vert C\Vert_1:=\trace{(\sqrt{C^\dagger C})}<\infty\rbrace\subseteq\mathcal K(\mathcal H),
\end{equation*}
which forms a Banach space under the trace norm $\|\cdot\|_1$ and
constitutes a two-sided ideal in the $C^*$-algebra of all bounded operators
$\mathcal B(\mathcal H)$.
Important properties of the trace are
\begin{align}\label{eq:4}
\operatorname{tr}\big((\langle x,\cdot\rangle y)T\big) = \langle x,Ty\rangle
\quad\text{and}\quad
|\operatorname{tr}(CT)|\leq\|C\|_1\Vert T\Vert_{\text{op}}
\end{align}
for all $x,y\in\mathcal H$, $C\in\mathcal B^1(\mathcal H)$ and $T\in\mathcal B(\mathcal H)$. Furthermore, the trace class has the approximation property:

\begin{lemma}[\cite{dirr_ve}, Lemma 3.2]\label{lemma_block_approx}
Let $C\in\mathcal B^1(\mathcal H)$ and let $(e_n)_{n\in\mathbb N}$ be any orthonormal basis of $\mathcal H$. 
For arbitrary $k\in\mathbb N$, let $\Pi_k:=\sum_{j=1}^k\langle e_j,\cdot\rangle e_j$ 
denote the orthogonal projection onto $\operatorname{span}\lbrace e_1,\ldots,e_k\rbrace$. Then the sequence 
of ``block approximations'' $(\Pi_nC\Pi_n)_{n\in\mathbb N}$ converges (in trace norm) to $C$, i.e. 
$$
\lim_{n\to\infty}\|C-\Pi_n C\Pi_n\|_1=0\,.
$$
\end{lemma}

\begin{defi}\label{def1}
\begin{itemize}
\item[(a)] A linear map $T:\mathcal B(\mathcal H)\to\mathcal B(\mathcal G)$ is trace-preserving if 
$T(\mathcal B^1(\mathcal H))\subseteq\mathcal B^1(\mathcal G)$ 
with $\operatorname{tr}(T(A))=\operatorname{tr}(A)$ for all $A\in\mathcal B^1(\mathcal H)$.
\item[(b)] A bi-stochastic quantum map is a linear, ultraweakly continuous\footnote{The 
ultraweak topology is the weak-$*$ topology on $\mathcal B(\mathcal H)$ inherited by the isometrically 
isomorphic map $\psi:\mathcal B(\mathcal H)\to(\mathcal B^1(\mathcal H))'$, 
$B\mapsto \operatorname{tr}(B(\cdot))$.\label{footnote_uw}}, completely positive, unital (identity-preserving)
and trace-pre\-serving map
$T:\mathcal B(\mathcal H)\to\mathcal B(\mathcal G)$. We define
\begin{align*}
\mathbb S (\mathcal H,\mathcal G):=
\lbrace T:\mathcal B(\mathcal H)\to\mathcal B(\mathcal G)\, |\, T\text{ is a bi-stochastic quantum map}\rbrace
\end{align*}
and $\mathbb S (\mathcal H):= \mathbb S (\mathcal H,\mathcal H)$.
\end{itemize}
\end{defi}
Thus using the terminology of \cite[Def.~2]{vE_dirr_semigroups}, a bi-stochastic quantum map is 
a Heisenberg quantum channel which also is trace-preserving and its restriction to the trace class is 
a Schr\"odinger quantum channel. Using \cite[Prop.~2]{vE_dirr_semigroups} this directly implies the following.
\begin{lemma}\label{lemma_qc_norm}
Let $T\in\mathbb S(\mathcal H,\mathcal G)$ and consider the restricted (well-defined) map $T_{\mathcal B^1}:\mathcal B^1(\mathcal H)\to\mathcal B^1(\mathcal G)$. Then $\|T\|_\text{op}=\|T_{\mathcal B^1}\|_\text{op}=1$ so
$$\|T(B)\|_\text{op}\leq\|B\|_\text{op}\quad\text{and}\quad \|T(A)\|_1\leq\| A\|_1$$
for all $B\in\mathcal B(\mathcal H)$, $A\in\mathcal B^1(\mathcal H)$.
\end{lemma}
\bigskip
\section*{Appendix \app{B}: Proof of von Neumann Type of Trace Inequality}


The following von Neumann type of trace (in-)equality was used in Sec.~\ref{sec:set_conv}

\thmskip

\noindent\textbf{Corollary \ref{cor_von_Neumann}.} 
\textit{Let $C\in\mathcal B^1(\mathcal H)$ and $T\in\mathcal K(\mathcal H)$ both be selfadjoint.
Then 
\begin{equation*}
\sup_{U\in\mathcal U(\mathcal H)}\operatorname{tr}(C U^\dagger TU)
= \sum\nolimits_{j=1}^\infty \lambda_j^\downarrow(C^+) \lambda_{j}^\downarrow(T^+) 
+ \sum\nolimits_{j=1}^\infty \lambda_j^\downarrow(C^-) \lambda_{j}^\downarrow(T^-)\,,
\end{equation*}
where $\lambda_j^\downarrow(C^+)$, $\lambda_{j}^\downarrow(T^+)$ and $\lambda_j^\downarrow(C^-)$, $\lambda_{j}^\downarrow(T^-)$
denote the decreasing eigenvalue sequences of the positive semi-definite operators $C^+$, $T^+$ and $C^-$, $T^-$, respectively, where $C=C^+-C^-$ and $T=T^+-T^-$ as usual.}

\thmskip

Below, we provide a proof of the above statement, which to the best of our knowledge is
new. To this end, we need the notion of set convergence using the Hausdorff 
metric on compact subsets (of $\mathbb C$) and the associated notion of convergence, see, e.g., 
\cite{nadler1978}. The distance between $z \in \mathbb C$ and any non-empty compact subset 
$A \subseteq \mathbb C$ is defined by
\begin{align}\label{eq.Hausdorff-1}
d(z,A) := \min_{w \in A} d(z,w) = \min_{w \in A} |z-w|\,.
\end{align}
Based on \eqref{eq.Hausdorff-1} the \emph{Hausdorff metric} $\Delta$ on the set of all non-empty
compact subsets of $\mathbb C$ is given by
\begin{align*}
\Delta(A,B) := \max\Big\lbrace \max_{z \in A}d(z,B),\max_{z \in B}d(z,A) \Big\rbrace\,.
\end{align*}
The following characterization of the Hausdorff metric is readily verified.

\begin{lemma}\label{lemma_11}
Let $A,B \subset \mathbb C$ be two non-empty compact sets and let $\varepsilon > 0$. 
Then $\Delta(A,B) \leq \varepsilon$ if and only if for all $z \in A$, there exists $w \in B$ 
with $d(z,w) \leq \varepsilon$ and vice versa.
\end{lemma}

With this metric one can introduce the notion of convergence for sequences $(A_n)_{n\in\mathbb N}$ 
of non-empty compact subsets of $\mathbb C$ such that the maximum-operator is continuous in the 
following sense.

\begin{lemma}\label{lemma_lim_max}
Let $(A_n)_{n\in\mathbb N}$ be a bounded sequence of non-empty, compact subsets of $\mathbb R$ 
which converges to $A\subset\mathbb R$. Then the sequence of real numbers $(\max A_n)_{n\in\mathbb N}$ 
is convergent with 
\begin{align*}
\lim_{n\to\infty}(\max A_n)=\max{}\big(\lim_{n\to\infty}A_n\big)=\max A\,.
\end{align*}
\end{lemma}

\begin{proof}
Let $\varepsilon>0$. By assumption, there exists $N\in\mathbb N$ such that $\Delta(A_n,A)<\varepsilon$ 
for all $n\geq N$. Hence, by Lemma \ref{lemma_11}, there exists $a_n\in A_n$ such that 
$|\max A - a_n| < \varepsilon$ and thus
\begin{align*}
\max A < a_n + \varepsilon < \max A_n + \varepsilon\,.
\end{align*}
Similarly, there exists $a \in A$ such that $|\max A_n - a| < \varepsilon$ and thus
\begin{align*}
\max A_n < a + \varepsilon < \max A + \varepsilon\,.
\end{align*}
Combining both estimates, we get $|\max A - \max A_n| < \varepsilon$.
\end{proof}

Just like \cite[Thm.~3.1]{dirr_ve} one can show the following:

\begin{lemma}\label{lemma_3}
Let $C\in\mathcal B^1(\mathcal H),T\in\mathcal B(\mathcal H)$  and $(C_n)_{n\in\mathbb N}$ be a 
sequence in $\mathcal B^1(\mathcal H)$ which converges to $C$ w.r.t.~$\|\cdot\|_1$. Then
$$
\lim_{n\to\infty}\overline{W_{C_n}(T)}=\overline{W_C(T)}\,.
$$
Moreover, if $T$ is compact as well, then
\begin{align*}
\lim_{k\to\infty}\overline{W_C(\Pi_kT\Pi_k)}=\overline{W_C(T)}\,,
\end{align*}
where $\Pi_k$ is the orthogonal projection onto the span of the first $k$ elements of 
an arbitrary orthonormal basis $(e_n)_{n\in\mathbb N}$ of $\mathcal H$.
\end{lemma}

\begin{proof}[Proof of Corollary \ref{cor_von_Neumann}]
Let $C\in\mathcal B^1(\mathcal H)$ and $T\in\mathcal K(\mathcal H)$ both be selfadjoint and 
let us first assume that $T$ has at most $k\in\mathbb N$ non-zero eigenvalues. Then
\begin{equation}\label{eq:finite_von_Neumann}
\max \operatorname{conv}(\overline{P_C(T)}) = 
\sum\nolimits_{j=1}^k\lambda_j^\downarrow(C^+) \lambda_{j}^\downarrow(T^+) 
+ \sum\nolimits_{j=1}^k \lambda_j^\downarrow(C^-) \lambda_{j}^\downarrow(T^-)\,,
\end{equation}
with $P_C(T)$ the $C$-spectrum of $T$ (cf.~Lemma \ref{lemma_4}), is
straightforward to show.

Now let us address the general case. Choose any orthonormal eigenbasis $(e_n)_{n\in\mathbb N}$ of $T$  with icorresponding modified eigenvalue sequence
(Lemma \ref{lemma_berb_4_6}). Moreover, let 
$\Pi_k=\sum\nolimits_{j=1}^k\langle e_j,\cdot\rangle e_j$ the projection onto the span
of the first $k$ eigenvectors of $T$. Then $\Pi_kT\Pi_k$ has at most $k$ non-zero eigenvalues and our preliminary considerations combined with Lemma \ref{lemma_4}, \ref{lemma_lim_max} and
\ref{lemma_3} readily imply
\begin{align*}
\sup_{U\in\mathcal U(\mathcal H)}&\operatorname{tr}(C U^\dagger TU)
= \max \overline{W_c(T)} = \max \lim_{k\to\infty} \overline{W_C(\Pi_kT\Pi_k)}\\
& =\lim_{k\to\infty}\max \overline{W_C(\Pi_kT\Pi_k)} 
= \lim_{k\to\infty}\max \operatorname{conv}(\overline{P_C(\Pi_kT\Pi_k)})\\
& =\lim_{k\to\infty}\Big(\sum\nolimits_{j=1}^k \lambda_j^\downarrow(C^+) \lambda_{j}^\downarrow(\Pi_kT^+\Pi_k)
+ \sum\nolimits_{j=1}^k \lambda_j^\downarrow(C^-) \lambda_{j}^\downarrow(\Pi_kT^-\Pi_k)\Big)\\
& = \sum\nolimits_{j=1}^\infty \lambda_j^\downarrow(C^+) \lambda_{j}^\downarrow(T^+) 
+ \sum\nolimits_{j=1}^\infty \lambda_j^\downarrow(C^-) \lambda_{j}^\downarrow(T^-)\,,
\end{align*}
where we used the identity $(\Pi_kT\Pi_k)^\pm = \Pi_kT^\pm\Pi_k$. This yields the result.
\end{proof}

\bigskip

\section*{Appendix \app{C}: Proof of the Main Theorem for Bounded $H_0$}

\noindent \textbf{Lemma \ref{lemma_maj_closed}.} \textit{Let $\rho_0\in\mathbb D(\mathcal H)$. Then the set 
$\lbrace \rho\in\mathbb D(\mathcal H)\,|\,\rho\prec\rho_0\rbrace$ is closed w.r.t.~the trace norm $\|\cdot\|_1$.}
\begin{proof}
For given $\omega\in \clon{\lbrace \rho\in\mathbb D(\mathcal H)\,|\,\rho\prec\rho_0\rbrace}$ 
there exists a sequence $(\rho_n)_{n\in\mathbb N}$ in 
$\lbrace \rho\in\mathbb D(\mathcal H)\,|\,\rho\prec\rho_0\rbrace\subseteq\mathbb D(\mathcal H)$ 
such that $\|\omega-\rho_n\|_1\to 0$ as $n\to\infty$. Obviously,
\begin{equation*}
1 = \lim_{n\to\infty}\operatorname{tr}(\rho_n)=\operatorname{tr}(\omega) 
\quad \text{and} \quad
0 \leq \lim_{n\to\infty}\langle x,\rho_nx\rangle=\langle x,\omega x\rangle
\quad\text{by Eq.~\eqref{eq:4}}
\end{equation*}
for all $x\in\mathcal H$ so $\omega\in\mathbb D(\mathcal H)$. Now let $T\in\mathcal K(\mathcal H)$ be 
arbitrary but selfadjoint. Then Lemma \ref{lemma_lim_max} and \ref{lemma_3} implies
\begin{align*}
K_\omega(T)=\max\overline{W_\omega(T)}&=\max\lim_{n\to\infty}\overline{W_{\rho_n}(T)}\\
&=\lim_{n\to\infty}\max\overline{W_{\rho_n}(T)}=\lim_{n\to\infty}K_{\rho_n}(T)\,.
\end{align*}
On the other hand, due to Prop.~\ref{thm_1} and $\rho_n\prec\omega$ for all $n\in\mathbb N$, one has
$$
K_\omega(T)=\lim_{n\to\infty}K_{\rho_n}(T)\leq K_{\rho_0}(T)
$$
which again by Prop.~\ref{thm_1} (as $T$ was chosen arbitrarily) implies $\omega\prec\rho_0$.
\end{proof}
\thmskip

\noindent\textbf{Lemma \ref{lemma_unitary_channel_approx}} (Unitary channel approximation)\textbf{.\;} 
\textit{Consider a subset $R\subseteq \mathcal U(\mathcal H)$ of the unitary group on $\mathcal H$ such 
that $\overline{R}^s=\mathcal U(\mathcal H)$, i.e.~its strong closure relative to 
$\mathcal U(\mathcal H)$ yields the full group. Furthermore let $\rho\in\mathbb D(\mathcal H)$ and 
$U\in\mathcal U(\mathcal H)$. Then for all $\varepsilon>0$ one can find $\tilde U\in R$ such that}
$
\|U\rho U^\dagger-\tilde U\rho\tilde U^\dagger\|_1<\varepsilon\,.  
$

\begin{proof}
Due to Lemma \ref{lemma_berb_4_6}
there exists a modified eigenvalue sequence $(r_n)_{n\in\mathbb N}\in\ell^1_+(\mathbb N)$ and an orthonormal 
basis $(e_n)_{n\in\mathbb N}$ of $\mathcal H$ such that $\rho = \sum_{j=1}^\infty r_j\langle e_j,\cdot\rangle e_j$.
Then one also finds $N\in\mathbb N$ such that the ``tail'' of $\rho$ is sufficiently small, i.e.
\begin{equation}\label{eq:str_approx_1}
\sum\nolimits_{j=N+1}^\infty r_j<\frac{\varepsilon}{6}
\quad\text{and}\quad \sum\nolimits_{j=1}^N r_j > 0\,.
\end{equation}
By assumption there is 
$\tilde U\in R\subseteq\mathcal U(\mathcal H)$ such that
\begin{equation}\label{eq:str_approx_2}
\|Ue_j-\tilde Ue_j\|_{\mathcal H}<\varepsilon/\big(6\sum\nolimits_{j=1}^N r_j\big)
\quad\text{for all $j=1,\ldots,N$}.
\end{equation}
Moreover, the triangle inequality, non-negativity of $r_n$ and the trace norm identity 
 $\|\langle x,\cdot\rangle y\|_1=\|x\|\|y\|$ for all $x,y\in\mathcal H$
imply
\begin{align*}
\|U &\rho U^\dagger - \tilde U\rho\tilde U^\dagger\|_1 \leq 
\|U\rho U^\dagger - U\rho\tilde U^\dagger\|_1 + \|U\rho \tilde U^\dagger-\tilde U\rho\tilde U^\dagger\|_1\\ 
& = \Big\| \sum_{j=1}^\infty r_j\langle Ue_j-\tilde Ue_j,\cdot\rangle Ue_j \Big\|_1
+ \Big\| \sum_{j=1}^\infty r_j\langle \tilde Ue_j,\cdot\rangle (Ue_j-\tilde Ue_j) \Big\|_1\\
&\leq  \sum\nolimits_{j=1}^\infty r_j \| Ue_j-\tilde Ue_j\|(\| Ue_j\|+\|\tilde Ue_j\|)\,.
\end{align*}
Splitting the sum at $N$ and using Eq.~\eqref{eq:str_approx_1} and \eqref{eq:str_approx_2} 
finally yields the estimate
$
\|U\rho U^\dagger-\tilde U\rho\tilde U^\dagger\|_1\leq 2\sum_{j=1}^{N} r_j \| Ue_j-\tilde Ue_j\|+4\sum_{j=N+1}^\infty r_j < \frac{\varepsilon}{3}+\frac{2\varepsilon}{3} =\varepsilon\;.
$
\end{proof}
Next, let us refine Lemma \ref{rem_trotter} in terms of precision as follows:
\begin{lemma}\label{lemma_trotter_approx}
Let $V\in\mathcal K(\mathcal H)\setminus\lbrace0\rbrace$ be normal, $H_0\in\mathcal B(\mathcal H)$ 
be selfadjoint, $\rho\in\mathbb D(\mathcal H)$ be arbitrary and $[0,T] \subset\mathbb R_0^+$ be given.
Furthermore let $R\subseteq \mathcal U(\mathcal H)$ with $\overline{R}^s=\mathcal U(\mathcal H)$,
where the closure is taken in $\mathcal U(\mathcal H)$. Then for all $\varepsilon>0$ there exists $m\in\mathbb N$ and $U_1,\ldots,U_m\in R$ such that for all $s \in [0,T]$
\begin{equation*}
\Big\|   \exp(-s\hGA_V)(\rho)-\prod_{j=1}^m \Big(\Adr_{U_j}\circ\exp\Big(  \frac{-is\hH_0-s\hGA_V}{m}  \Big) \Big)(\rho) \Big\|_1<\varepsilon
\end{equation*}
\end{lemma}
\begin{proof}
By Lemma \ref{rem_trotter} there exists $m\in\mathbb N$ with
\begin{equation}\label{eq:approx_3}
\Big\Vert \exp(-s\hGA_V)-\Big( \underbrace{\exp\Big(\frac{is\hH_0}{m}\Big)}_{=:F}\circ\underbrace{\exp\Big(\frac{-is\hH_0-s\hGA_V}{m}\Big) }_{=:G}\Big)^m\Big\Vert_\text{op}<\frac{\varepsilon}{2}\,,
\end{equation}
for all $s \in [0,T]$, where $F$ is a unitary channel and $G$ is unital 
(because $V$ is normal) and reflects the operator solution of \eqref{eq:GA} with $u_1(s)= \dots = u_m(s)=0$
and $\gamma(s)=1$, i.e.~the noisy but uncontrolled evolution of the system.

\medskip
For convenience define $\rho_j:=(G\circ (F\circ G)^{m-j})(\rho)$ for $j=1,\ldots,m$. Then, 
Lemma \ref{lemma_unitary_channel_approx} yields $U_j\in R\subseteq\mathcal U(\mathcal H)$ with
$
\| F(\rho_j)-U_j\rho_jU_j^\dagger\|_1 < \frac{\varepsilon}{2m}
$
for all\footnote{Note that $\rho_j = \rho_j(s)$ depends on $s \in [0,T]$ as $F$ and $G$ do so.
Moreover, the set $\{\rho_j(s)\,|\, s \in [0,T]\}$ is compact as $F$ and $G$ are continuous in $s$ and 
hence the proof of Lemma \ref{lemma_unitary_channel_approx} can be easily modified to obtain the desired result.} 
$s \in [0,T]$. Finally, Lemma \ref{lemma_qc_norm} and Lemma \ref{lem:telescope} (below) imply
 \begin{align*}
&\Big\|  \exp(-s\hGA_V)(\rho)-\prod_{j=1}^m (\Adr_{U_j}{}\!\circ g )  (\rho)\Big\|_1\\[1mm]
& \leq \Big\|   \exp(-s\hGA_V)-(F\circ G)^m \Big\|_\text{op}\| \rho\|_1
+ \Big\| (F \circ G)^m (\rho)-\prod_{j=1}^m (\Adr_{U_j}{}\!\circ G ) (\rho)\Big\|_1\\
& <\frac{\varepsilon}{2} + \sum_{j=1}^{m} \Big\| \prod_{k=1}^{j-1}(\Adr_{U_k}{}\!\circ G) \circ (F-\Adr_{U_j}) \circ \underbrace{G \circ (F \circ G)^{m-j} (\rho)}_{=\rho_j}\Big\|_1\\
&\leq \frac{\varepsilon}{2} + \sum_{j=1}^{m}  \underbrace{\Big( \prod_{k=1}^{j-1} \|  \Adr_{U_k} {} \|_\text{op}\, \| G\|_\text{op} \Big)}_{=1}     \big\|F(\rho_j)-U_j\rho_jU_j^\dagger\big\|_1
<\frac{\varepsilon}{2}+\frac{\varepsilon}{2}=\varepsilon\,. 
\end{align*}
for all $s \in [0,T]$.
\end{proof}

A simple and readily verified induction argument shows:
\begin{lemma}\label{lem:telescope}
Let $m\in\mathbb N$ and $A_1,\ldots,A_m$, $B_1,\ldots,B_m: D \to D$ be arbitrary maps acting on some common 
domain $D$ be given. Then
$$
\prod_{j=1}^m \,A_j - \prod_{j=1}^m\, B_j 
= \sum_{j=1}^m\Big( \prod_{k=1}^{j-1} A_k \circ (A_j-B_j) \circ \prod_{k=j+1}^m B_j \Big)\,.
$$
Here and henceforth, the order of the ``product'' $\prod_{j=1}^m \,A_j$ shall be fixed by
$A_1 \circ \cdots \circ A_m$.
\end{lemma}

\begin{proof}[{\bf Proof of Theorem \ref{thm_normal_V}}]{$\hfill$}

``$\subseteq$'': Obviously, $\rho_0\in\lbrace\rho\in\mathbb D(\mathcal H)\,|\, \rho\prec\rho_0 \rbrace$ and 
by assumption of $V$ being normal,
$
\hGA(\mathbf{1})=-VV^\dagger+\frac12(V^\dagger V+V^\dagger V)=0
$. 
Thus the operator solution of \eqref{eq:GA} remains in $\mathbb S(\mathcal H)$ for $t \geq 0$, and 
by Lemma \ref{lemma_li} the set $ \lbrace\rho\in\mathbb D(\mathcal H)\,|\, \rho\prec\rho_0 \rbrace$ 
is forward invariant, i.e.~solutions of the given control problem can never leave the set of states 
majorized by $\rho_0$. Taking the $\|\cdot\|_1$-closure by Lemma \ref{lemma_maj_closed} yields
$$
\clon{\reach_{\Sigma_V}(\rho_0)}
	\subseteq 
	\clon{\lbrace \rho\in\mathbb D(\mathcal H)\,|\,\rho\prec\rho_0\rbrace}
	=\lbrace \rho\in\mathbb D(\mathcal H)\,|\,\rho\prec\rho_0\rbrace\,.
$$

``$\supseteq$'': As $V\in\mathcal K(\mathcal H)$ is normal, by Lemma \ref{lemma_berb_4_6} there exists an 
orthonormal basis $(f_j)_{j\in\mathbb N}$ of $\mathcal H$ such that $V=\sum_{j=1}^\infty v_j\langle f_j,\cdot\rangle f_j$ 
with modified eigenvalue sequence $(v_j)_{j\in\mathbb N}$. Whenever we use the term ``diagonal'' 
or ``diag'' in the following, it always refers to $(f_j)_{j\in\mathbb N}$.

\medskip

Let $\varepsilon>0$ and $\rho\in\mathbb D(\mathcal H)$ with $\rho\prec\rho_0$ be given. We now have to find 
$\rho_F\in \reach_{\Sigma_V}(\rho_0)$ such that $\|\rho-\rho_F\|_1<\varepsilon$. As seen before there exist 
$x,y\in\ell^1_+(\mathbb N)$, $x,y\neq 0$ as well as unitary $W_1,W_2\in\mathcal B(\mathcal H)$ such that
\begin{equation}\label{eq:W_1-W_2}
\rho=W_1\operatorname{diag}(x)W_1^\dagger\quad\text{and}\quad \rho_0=W_2\operatorname{diag}(y)W_2^\dagger
\end{equation}
with $x\prec y$ (so $x$ and $y$ denote the modified eigenvalue sequence of $\rho$
and $\rho_0$, respectively, see also Remark~\ref{rem_maj}). 
\medskip

First assume that the original and the modified eigenvalue sequence of $\rho$ as well as $\rho_0$ co{\"i}ncide, i.e.~$x=x^\downarrow$, $y=y^\downarrow$ from the start (necessary to apply Prop.~\ref{prop_schur_horn_gohberg}).
The subsequent steps of the proof were sketched in the main text on page \pageref{eq:steps_idea}, 
where {\bf Step 1 \& 3} are the mere application of a suitable unitary channel whereas 
{\bf Step 2} is about (approximately) getting rid of all ``off-diagonal'' elements 
$\langle f_j,UW_2^\dagger\rho_0 W_2U^\dagger f_k\rangle$ of 
$UW_2^\dagger\rho_0 W_2U^\dagger = U\operatorname{diag}(y)U^\dagger$.

\medskip
{\bf Step 1:} By assumption $\Sigma_0$ is strongly operator controllable so we have the 
unitary orbit of $\rho_0$ in the 
closure of $\reach_{\Sigma_V}(\rho_0)$. 
Although we may not have access to $\Adr_{UW_2^\dagger}=UW_2^\dagger (\cdot) W_2U^\dagger$ directly, 
by Lemma \ref{lemma_unitary_channel_approx} we find $\tilde U\in\mathcal B(\mathcal H)$ unitary 
such that $\tilde U\rho_0\tilde U^\dagger\in \reach_{\Sigma_V}(\rho_0)$ with
$$
\|U\operatorname{diag}(y)U^\dagger-\tilde U\rho_0\tilde U^\dagger\|_1=\|UW_2^\dagger \rho_0 W_2U^\dagger-\tilde U\rho_0\tilde U^\dagger\|_1<\varepsilon / 3\,.
$$

{\bf Step 2:} By Lemma \ref{lemma_normal_generator} the pure noise generator $\hGA_V$ acts like
\begin{align}
|\langle f_j,\exp(-t\hGA_V)&(X) f_k\rangle|=\Big|\exp\Big(-\frac{t|v_j-v_k|^2}2\Big)\exp(it\operatorname{Im}(v_j\overline{v_k}))\langle f_j,Xf_k\rangle\Big|\nonumber\\
&=\exp\Big(-\frac{t|v_j-v_k|^2}2\Big)|\langle f_j,Xf_k\rangle|\leq |\langle f_j,Xf_k\rangle|\label{eq_approx_3a2}
\end{align}
on arbitrary $X\in\mathcal B(\mathcal H)$ for all $j,k\in\mathbb N$ and $t\in\mathbb R_0^+$. Evidently,
\begin{equation*}
\lim_{t\to\infty}\langle f_j,\exp(-t\hGA_V)(X) f_k\rangle=\begin{cases} 0 & \text{if }v_j\neq v_k\,, \\ 
\langle f_j,Xf_k\rangle & \text{else}\,. \end{cases}
\end{equation*}
If we assume $v_j\neq v_k$ for all $j\neq k$, all the off-diagonal terms of $X$ vanish in the limit $t\to\infty$ 
and one is left with $\sum_{j=1}^\infty \langle f_j,Xf_j\rangle \langle f_j,\cdot\rangle f_j=:P(X)$. Note that this 
projection map has Kraus operators $(\langle f_j,\cdot\rangle f_j)_{j=1}^\infty$ so $P\in \mathbb S(\mathcal H)$. 
Since we want to approximate a density operator in the trace norm, we only have to care about 
a sufficiently large upper left block of the matrix representation $(\langle f_j,Xf_k\rangle)_{j,k \in \mathbb N}$ 
as the rest is ``already small'' in the trace norm. More formally, by Lemma \ref{lemma_block_approx} there exists 
$N_1\in\mathbb N$ such that 
\begin{equation}\label{eq_approx_4a}
\|\tilde U\rho_0\tilde U^\dagger-\Pi_n \tilde U\rho_0\tilde U^\dagger\Pi_n\|_1<\varepsilon / 24
\end{equation}
for all $n\geq N_1$, where  $\Pi_n := \sum_{j=1}^n\langle f_j,\cdot\rangle f_j$ for all 
$n\in\mathbb N$. \medskip

Of course, there is no reason for the eigenvalues of $V$ to be pairwise different. Therefore we have to make sure 
that the upper left block is large enough such that it corresponds to at least two different 
eigenvalues of $V$---thus we have access to partial decoherence, which we then may spread anywhere needed via 
permutation channels.\medskip

Due to $V\neq 0$ and $v_j\to 0$ as $j\to\infty$ (compactness of $V$), there exists $M\in\mathbb N$ such that 
$v_1\neq v_M$. On the other hand \eqref{eq_approx_4a} still holds so we define $N:=\max\lbrace N_1,M\rbrace$. 
Then, by construction and \eqref{eq_approx_3a2}, we know that $\langle f_1,Xf_M\rangle$ (and 
$\langle f_M,Xf_1\rangle$) tend to zero when pure noise is applied.


Thus we find $\alpha\in\mathbb N_0$, $\alpha\leq N(N-1)/2$ (number of matrix elements above the diagonal), 
permutation operators $\sigma_1,\ldots,\sigma_\alpha\in\mathcal U(\mathcal H)$ (in abuse of notation 
we write $f_j\mapsto \sigma_lf_j=f_{\sigma^{-1}_l(j)}$, yet the explicit form of $\sigma_j$ is not that important) 
and relaxation times $s_1,\ldots,s_\alpha\in\mathbb R_0^+$ such that
\begin{itemize}
\item the permutations only operate non-trivially on the $N\times N$-block, i.e.~for all $l=1,\ldots,\alpha$ 
and $k>N$ one has $\sigma_lf_k=f_k$.
\item
for every matrix element $\langle f_k,\cdot\rangle f_j$ with $j,k=1,\ldots,N$, $j\neq k$ there
exists a permutation $\sigma_l$ with $1\leq l \leq \alpha$ such that $\langle f_k,\cdot\rangle f_j$ sits 
in the ``relaxation'' spot (i.e.~$\langle f_1,\cdot\rangle f_M$ or $\langle f_M,\cdot\rangle f_1$).
More precisely, 
\begin{equation}\label{eq:perm_relax}
\big\|\operatorname{Ad}_{\sigma_l^\dagger} \circ \exp(-s_l\hGA_V) \circ \operatorname{Ad}_{\sigma_l}
\big(\langle f_k,\cdot\rangle f_j\big)\big\|_1\leq \frac{\varepsilon}{12N^2} \,.
\end{equation}
\item after having successively applied all operations from \eqref{eq:perm_relax}, every matrix element $\langle f_k,\cdot\rangle f_j$ is in its 
original spot because all $\langle f_k,\cdot\rangle f_j$ are eigenvectors of $\exp(-s_l\hGA_V)$.
\end{itemize}
Now, using linearity of the involved maps, the estimate in question reads
\begin{align*}
\| P( &\tilde U\rho_0 \tilde U^\dagger)-\prod_{m=1}^{\alpha}\big( \operatorname{Ad}_{\sigma_m^\dagger}{}\!\circ \exp(-s_m\hGA_V) \circ \operatorname{Ad}_{\sigma_m} \big)( \tilde U\rho_0 \tilde U^\dagger)\|_1<\\
& \Big( \|P\|_\text{op} + \prod_{m=1}^\alpha \| \operatorname{Ad}_{\sigma_m} \|^2_\text{op}\, \|  \exp(-s_m\hGA_V)  \|_\text{op} \Big) \| \tilde U\rho_0 \tilde U^\dagger-\Pi_N  \tilde U\rho_0 \tilde U^\dagger\Pi_N\|_1   \\
+ &\Big\| \sum_{j,k=1}^{N} \langle f_j,\tilde U\rho_0 \tilde U^\dagger f_k\rangle\Big( P-\prod_{m=1}^{\alpha} \operatorname{Ad}_{\sigma_m^\dagger}{}\!\circ \exp(-s_m\hGA_V) \circ \operatorname{Ad}_{\sigma_m}  \Big)(\langle f_k,\cdot\rangle f_j)\Big\|_1
\end{align*}
The first summand is smaller than $2\cdot \frac{\varepsilon}{24}=\frac{\varepsilon}{12}$ by Lemma 
\ref{lemma_qc_norm} and \eqref{eq_approx_4a}. For the second one notice that
\begin{align*}
\Big( P-\prod_{m=1}^{\alpha} \operatorname{Ad}_{\sigma_m^\dagger}{}\!\circ \exp (-s_m\hGA_V) \circ \operatorname{Ad}_{\sigma_m} \Big)(\langle f_j,\cdot\rangle f_j) = 0
\end{align*}
for all $j\in\mathbb N$. 
Now, $P(\langle f_k,\cdot\rangle f_j)=0$ whenever $j\neq k$. Moreover, 
\begin{align*}
\Big\|\prod_{m=1}^{\alpha} \big(\operatorname{Ad}_{\sigma_m^\dagger}{}\!\circ \exp(-s_m\hGA_V) \circ \operatorname{Ad}_{\sigma_m} \big)(\langle f_k,\cdot\rangle f_j)\Big\|_1 \leq\frac{\varepsilon}{12N^2}\,.
\end{align*}
by 
\eqref{eq_approx_3a2} and \eqref{eq:perm_relax}. 
Putting 
together gives the estimate
\begin{align*}
\Big\| P&(\tilde U\rho_0 \tilde U^\dagger)- \prod_{m=1}^{\alpha}\big( \operatorname{Ad}_{\sigma_m^\dagger}{}\!\circ \exp(-s_m\hGA_V) \circ \operatorname{Ad}_{\sigma_m} \big)( \tilde U\rho_0 \tilde U^\dagger)\Big\|_1\\
< & \frac{\varepsilon}{12}+ \sum_{\substack{j,k=1\\ j\neq k}}^{N}\underbrace{|\langle f_j,\tilde U\rho_0 \tilde U^\dagger f_k\rangle|}_{\leq 1}\Big\|\prod_{m=1}^{\alpha} \big(\operatorname{Ad}_{\sigma_m^\dagger}{}\!\circ \exp(-s_m\hGA_V) \circ \operatorname{Ad}_{\sigma_m} \big)(\langle f_k,\cdot\rangle f_j)\Big\|_1\\
< & \frac{\varepsilon}{12} + \sum\limits_{j,k=1,\, j\neq k}^n \frac{\varepsilon}{12N^2}\leq \frac{\varepsilon}{6}\,.
\end{align*}
This leaves us with two problems:
\begin{itemize}
\item[1.] We have to approximate all permutation channels.
\item[2.] We do not have access to pure noise $(\exp(-t\hGA_V))_{t\in\mathbb R_0^+}$ within the given control problem.
\end{itemize}
For solving the first problem we exploit that we can strongly approximate every unitary 
channel. First, to simplify the upcoming computations, let us assume w.l.o.g.~that $\sigma_\alpha$
is the identity and let us introduce the notation $\pi_l := \sigma_l \circ \sigma_{l-1}^\dagger$ for 
$l\in\lbrace 2,\ldots,\alpha\rbrace$ and $\pi_1 := \sigma_1$. Moreover, define
$$
\omega_l:=\big( \exp(-s_l\hGA_V)  \circ \prod_{m=l+1}^\alpha (\operatorname{Ad}_{\pi_m^\dagger}\circ \exp(-s_m\hGA_V) \big)(\tilde U\rho_0\tilde U^\dagger)\in\mathbb D(\mathcal H)
$$
for every $l\in\lbrace1,\ldots,\alpha\rbrace$ as then by Lemma \ref{lemma_unitary_channel_approx}
we find $\tilde\pi_l\in\mathcal U(\mathcal H)$ which we have access to within $\reach_{\Sigma_0}$ (and thus $\reach_{\Sigma_V}$) such that
\begin{equation}\label{eq:sigma_approx}
\|\tilde\pi_l^\dagger \omega_l \tilde\pi_l - \pi_l^\dagger \omega_l \pi_l\|_1 < \frac{\varepsilon}{12\alpha}\,.
\end{equation}
Then a telescope argument (cf.~Lemma \ref{lem:telescope})
yields the estimate
\begin{align*}
\Big\|\Big(&\prod_{m=1}^{\alpha}\big( \operatorname{Ad}_{\tilde\pi_m^\dagger}{}\!\circ \exp(-s_m\hGA_V) \big)-\prod_{m=1}^{\alpha}\big( \operatorname{Ad}_{\pi_m^\dagger}{}\!\circ \exp(-s_m\hGA_V)  \big)\Big)(\tilde U\rho_0\tilde U^\dagger)    \Big\|_1\\
&\leq \sum_{m=1}^\alpha \Big\| \Big(\prod_{l=1}^{m-1}( \operatorname{Ad}_{\tilde\pi_l^\dagger}{}\!\circ \exp(-s_l\hGA_V) )\circ ( \operatorname{Ad}_{\tilde\pi_m^\dagger}- \operatorname{Ad}_{\pi_m^\dagger})\Big)(\omega_m )\Big\|_1\\
&\leq\sum_{m=1}^\alpha\Big( \prod_{l=1}^{m-1}\| \operatorname{Ad}_{\tilde\pi_l^\dagger}{}\|_\text{op}\, \|\exp(-s_l\hGA_V)\|_\text{op}\Big) \|\tilde\pi_m^\dagger \omega_m \tilde\pi_m-\pi_m^\dagger \omega_m \pi_m\|_1
<\frac{\varepsilon}{12}
\end{align*}
where in the last step we once again used Lemma \ref{lemma_qc_norm}.

For the second problem we luckily may approximate the pure noise as precisely as needed using
Lemma \ref{lemma_trotter_approx}. For every $l=1,\ldots,\alpha$ define 
$$
\rho_l:=\prod_{m=l+1}^\alpha (\operatorname{Ad}_{\tilde\pi_m^\dagger}{}\!\circ \exp(-s_m\hGA_V))(\tilde U\rho_0\tilde U^\dagger)\in\mathbb D(\mathcal H)\,.
$$ 
Then by Lemma \ref{lemma_trotter_approx} there exists a {\sc cptp} map $F_l$ which we have access to
such that $\|\exp(-s_l\Gamma_V)(\rho_l) - F_l(\rho_l)\|_1<\frac{\varepsilon}{12\alpha}$. Just as before
\begin{align*}
\Big\|\big( \prod_{m=1}^{\alpha}\big( \operatorname{Ad}_{\tilde\pi_m^\dagger}{}\!\circ \exp(-s_m\hGA_V)  \big)-\prod_{m=1}^{\alpha}\big( \operatorname{Ad}_{\tilde\pi_m^\dagger}{}\!\circ \,F_m \big)\big)(\tilde U\rho_0\tilde U^\dagger)   \Big\|_1<\frac{\varepsilon}{12}\,.
\end{align*}

{\bf Step 3:} The current state $\tilde\rho:=\prod_{m=1}^{\alpha}\big( \operatorname{Ad}_{\tilde\pi_m^\dagger}{}\!\circ \,F_m \big)(\tilde U\rho_0\tilde U^\dagger)$ of the system is ``close to $\operatorname{diag}(x)$'' in 
the trace distance as we saw before. Now we want to apply the unitary channel generated by $W_1$ so
again by Lemma \ref{lemma_unitary_channel_approx} one finds unitary $\tilde W\in\mathcal B(\mathcal H)$
such that 
$
\|W_1\tilde\rho W_1^\dagger-\tilde W\tilde\rho\tilde W^\dagger\|_1<\frac{\varepsilon}{3}\,.
$
Then one has
$\rho_F = \operatorname{Ad}_{\tilde W}\circ\prod_{m=1}^{\alpha}\big( \operatorname{Ad}_{\tilde\pi_m^\dagger}{}\!\circ \,F_m\big)(\tilde U\rho_0\tilde U^\dagger)\in \reach_{\Sigma_V}(\rho_0)$ and by \eqref{eq:W_1-W_2}
\begin{align*}
\|\rho-\rho_F\|_1 & \leq \|W_1P(UW_2^\dagger\rho_0 W_2U^\dagger) W_1^\dagger-W_1P(\tilde U\rho_0\tilde U^\dagger) W_1^\dagger\|_1\\
&\quad+\|W_1P(\tilde U\rho_0\tilde U^\dagger) W_1^\dagger-W_1\tilde\rho W_1^\dagger\|_1+\|W_1\tilde\rho W_1^\dagger-\rho_F\|_1\,.
\end{align*}
As all channels involved are in $\mathbb S(\mathcal H)$, by Lemma \ref{lemma_qc_norm} we 
ultimately obtain
\begin{align*}
\|\rho-&\rho_F\|_1\leq \|\Adr_{W_1}\|_\text{op} \|P\|_\text{op} \|UW_2^\dagger\rho_0 W_2U^\dagger-\tilde U\rho_0\tilde U^\dagger\|_1\\
&\quad+\|\Adr_{W_1}\|_\text{op}\Big\|P(\tilde U\rho_0\tilde U^\dagger) -\prod_{m=1}^{\alpha}\big( \operatorname{Ad}_{\pi_m^\dagger}{}\!\circ \exp(-s_m\hGA_V)  \big)(\tilde U\rho_0\tilde U^\dagger)\,\Big\|_1\\
&\quad+ \|\Adr_{W_1}\|_\text{op} \Big\|\big(\prod_{m=1}^{\alpha}\big( \operatorname{Ad}_{\pi_m^\dagger}{}\!\circ \exp(-s_m\hGA_V) \big)-\\
&\hspace*{126pt}-\prod_{m=1}^{\alpha}\big( \operatorname{Ad}_{\tilde\pi_m^\dagger}{}\!\circ \exp(-s_m\hGA_V)  \big)\big)(\tilde U\rho_0\tilde U^\dagger)  \Big\|_1\\
&\quad+\|\Adr_{W_1}\|_\text{op}\Big\|\prod_{m=1}^{\alpha}\big( \operatorname{Ad}_{\tilde\pi_m^\dagger}{}\!\circ \exp(-s_m\hGA_V)  \big)(\tilde U\rho_0\tilde U^\dagger) -\tilde\rho\,\Big\|_1\\
&\quad+\|W_1\tilde\rho W_1^\dagger-\rho_F\|_1<\frac{\varepsilon}{3}+\frac{\varepsilon}{6}+\frac{\varepsilon}{12}+\frac{\varepsilon}{12}+\frac{\varepsilon}{3}=\varepsilon\,.
\end{align*}
Now what happens if we cannot apply Prop.~1 directly, i.e.~if the original and the modified eigenvalue sequence of $\rho=W_1\operatorname{diag}(x)W_1^\dagger$ or $\rho_0=W_2\operatorname{diag}(y)W_2^\dagger$ do not co{\"i}ncide? Given $\varepsilon>0$, we first of all find $N\in\mathbb N$ such that
\begin{equation}\label{eq:ineq_x_y}
\sum\nolimits_{j=N+1}^\infty x_j^\downarrow <\frac{\varepsilon}{12}\qquad \sum\nolimits_{j=N+1}^\infty y_j^\downarrow <\frac{\varepsilon}{12}\,.
\end{equation}
Take unitaries $ X$, $ Y\in\mathcal B(\mathcal H)$ so that
$
 X\rho X^\dagger=\operatorname{diag}(x_1^\downarrow,\ldots,x_N^\downarrow,*,*,\ldots)$ and $  Y\rho_0 Y^\dagger=\operatorname{diag}(y_1^\downarrow,\ldots,y_N^\downarrow,*,*,\ldots)
$
where the diagonal entries differ from the original only by a permutation on a finite block. As the tail of these new diagonals is ``already small'' we may change these elements within the realm of approximation. Given $\sum_{j=1}^N x_j^\downarrow\leq \sum_{j=1}^N y_j^\downarrow$ (because $\rho\prec\rho_0$) where this inequality may or may not be strict, we want to fill up $ X\rho X^\dagger$ with small entries such that the traces match. Define $\varphi:=\sum_{j=1}^N(y_j^\downarrow-x_j^\downarrow)$ where $0\leq\varphi <\frac{\varepsilon}{12}$ due to \eqref{eq:ineq_x_y} and $\rho\geq 0$, as well as $m:=\lceil \varphi/x_k^\downarrow\rceil\in\mathbb N$. Here $k\in\lbrace 1,\ldots,N\rbrace$ is chosen such that $x_k^\downarrow$ is the smallest non-zero entry of $(x_1^\downarrow,\ldots,x_N^\downarrow)$. The new (eigenvalue) sequences then are
$
\hat x:=(x_1^\downarrow,\ldots,x_k^\downarrow,\tfrac{\varphi}{m},\ldots,\tfrac{\varphi}{m},0,0,\ldots)$ (where $\varphi/m$ occurs $m$ times) and $\hat y:=(y_1^\downarrow,\ldots,y_N^\downarrow,0,0,\ldots)\,.
$
These sequences satisfy $\hat x^\downarrow=\hat x$, $\hat y^\downarrow=\hat y$ and $\hat x\prec\hat y$ (for this note that if $k<N$ then majorization forces $\sum_{j=1}^k x_j^\downarrow=\sum_{j=1}^k y_j^\downarrow=1$ and thus $\varphi=0$) so we could apply Prop.~1 to them. Now to $$
\omega:=\frac{\operatorname{diag}(\hat x)}{\sum_{j=1}^N y_j^\downarrow}\quad\text{ and }\quad\omega_0:=\frac{\operatorname{diag}(\hat y)}{\sum_{j=1}^N y_j^\downarrow}	\,,
$$
which are both in $\mathbb D(\mathcal H)$, we can apply the original scheme which yields a {\sc cptp} map $f$ on $\mathcal H$ such that $f(\omega_0)=\omega_F\in\mathfrak{reach}_{ \Sigma_V}(\omega_0)$ and $\|\omega-\omega_F\|_1<\frac{\varepsilon}{6}$. Of course linearity implies $\|\operatorname{diag}(\hat x)-f(\operatorname{diag}(\hat y))  \|_1<\frac{\varepsilon}{6}$. The final scheme goes as follows:
$$
\rho_0\overset{ Y}\longrightarrow Y\rho_0 Y^\dagger\approx\operatorname{diag}(\hat y)\overset{f}\longrightarrow \operatorname{diag}(\hat x)\approx X\rho X^\dagger\overset{ X^\dagger}\longrightarrow \rho
$$
More precisely, by Lemma \ref{lemma_unitary_channel_approx} we find unitaries $\tilde X,\tilde Y\in\mathcal B(\mathcal H)$ such that 
$$
\| Y\rho_0 Y^\dagger-\tilde Y\rho_0\tilde Y^\dagger\|_1<\tfrac{\varepsilon}{4},\quad \|\tilde X^\dagger   (f\circ\operatorname{Ad}_{\tilde Y})(\rho_0)\tilde X- X^\dagger   (f\circ\operatorname{Ad}_{\tilde Y})(\rho_0)    X\|_1<\tfrac{\varepsilon}{4}
$$
and $\rho_F:=(\operatorname{Ad}_{\tilde X^\dagger}\circ f\circ \operatorname{Ad}_{\tilde Y})(\rho_0) \in\mathfrak{reach}_{ \Sigma_V}(\rho_0)$. Putting things together,
\begin{align*}
\|\rho-\rho_F\|_1&\leq \|\rho- X^\dagger\operatorname{diag}(\hat x) X\|_1+\| X^\dagger\operatorname{diag}(\hat x) X- X^\dagger f(\operatorname{diag}(\hat y)) X  \|_1\\
&\hphantom{\leq \|\rho- X^\dagger\operatorname{diag}(\hat x) X\|_1}\ +\| X^\dagger f(\operatorname{diag}(\hat y)) X  -\rho_F \|_1\\
&<\frac{\varepsilon}{6}+\|\operatorname{Ad}_{ X^\dagger}\|_{\text{op}}\frac{\varepsilon}{6}+\|(\operatorname{Ad}_{ X^\dagger}\circ f)(\operatorname{diag}(\hat y))  -\rho_F \|_1\\	
&\leq\frac{\varepsilon}{3}+\|\operatorname{Ad}_{ X^\dagger}\|_{\text{op}}\|f \|_{\text{op}}\| \operatorname{diag}(\hat y)- Y\rho_0 Y^\dagger \|_1\\
&\hphantom{\leq\frac{\varepsilon}{4}}\ +\|\operatorname{Ad}_{ X^\dagger}\|_{\text{op}}\|f \|_{\text{op}} \| Y\rho_0 Y^\dagger-\tilde Y\rho_0\tilde Y^\dagger\|_1\\
&\hphantom{\leq\frac{\varepsilon}{4}}\ +\|  (\operatorname{Ad}_{  X^\dagger}\circ f\circ \operatorname{Ad}_{ \tilde Y})(\rho_0) -\rho_F \|_1<\varepsilon
\end{align*}
so $\rho\in\overline{\mathfrak{reach}_{ \Sigma_V}(\rho_0)}^1$, which concludes the proof.
\end{proof}

\bigskip

\section*{Appendix \app{D}: Theorem 1 for Unbounded Drift $H_0$}

The physically relevant case of an unbounded system Hamiltonian~$H_0$ is a bit more intricate.
To show that Eq.~\eqref{eq:GA} is well-defined in this general setting, we have to
resort to some basic results from the theory of strongly continuous one-parameter semigroups
as presented in \cite[Ch.~1.9 \& Ch.~5.5]{Davies76}.

To begin with, for selfadjoint $H_0,\ldots,H_m \in \mathcal B(\mathcal H)$ and 
$u_1, \dots, u_m \in \mathbb R$ and setting $H := H_0 + \sum_{j=0}^m u_j H_j$ and $\hH := \ad_H$,
the solution of
\begin{equation}\label{eq:evol_closed_sys}
\dot\rho(t)=-i\hH(\rho(t))\qquad  \rho(0)=\rho_0\in\mathbb D(\mathcal H)\,,
\end{equation}
is obviously given by 
applying the corresponding unitary channel 
\begin{equation*}
\rho(t)=e^{-itH}\rho_0\, e^{itH}
\end{equation*}
for all $t\in\mathbb R_+$ (even for all $t\in\mathbb R$). The control case $H(t)$ with 
piecewise constant control amplitudes $u_j(t)$  
is solved by compositions of such solutions.

Now let us assume that $H_0$ is unbounded and defined on some dense domain $D(H_0)\subset \mathcal H$. 
Then $H := H_0 + \sum_{j=0}^m u_j H_j$ is selfadjoint with dense domain $D(H) = D(H_0)$ 
and Stone's  Theorem implies that $-iH$ is the infinitesimal generator of a strongly continuous group 
$U(t) = e^{-itH}$ of unitary operators. The corresponding one-parameter group 
$\hU(t):= \Adr_{U(t)} := e^{-itH}(\cdot)\, e^{itH}$ of unitary channels (isometries\hspace{.3mm}\/!) is strongly 
continuous on the trace class $\mathcal B^1(\mathcal H)$ and hence it is generated via the 
densely defined, closed operator $-i\hH$. More precisely, one has the following result\footnote{NB: 
Davies~\cite{Davies76} proved the sequel on the Banach space of all selfadjoint trace-class operators. 
As selfadjointness is neither used nor necessary, we extend the results to $\mathcal B^1(\mathcal H)$.} 
which allows us to identify $-i\hH$ with $=-i\ad_H$.
\begin{lemma}[\cite{Davies76}, Ch.~5, Lemma 5.1]
The domain $D(\hH)$ of $\hH$ is the set of all $\rho\in\mathcal B^1(\mathcal H)$ such that $\rho(D(H))
\subset D(H)$ and such that the operator $H\rho - \rho H$ on $D(H)$ is norm bounded with an
extension to a trace class operator on $\mathcal H$. Moreover, one has the identity 
$\hH(\rho) = H\rho - \rho H =: \ad_H(\rho)$.
\end{lemma}

In the sequel, the explicit form of $D(\hH)$ is irrelevant as is the explicit construction
of $\hU_t$, e.g.,~via the Post-Widder inversion formula
$$
U(t)(\rho)=\lim_{n\to\infty}\Big(\mathbf 1-\frac{-it\hH}{n}\Big)^{-n}(\rho )
$$ 
for all $t\geq 0$ (even for all $t\in\mathbb R$) and $\rho\in\mathcal B^1(\mathcal H)$. Recall that the ODE 
\eqref{eq:evol_closed_sys} has a classical solution only for a dense set of initial values, namely for 
$\rho_0 \in D(\hH)\subset\mathcal B^1(\mathcal H)$. Nevertheless, we will write $U(t) = e^{-it\hH}$ although
this only holds formally.

\begin{lemma}\label{lemma_open_dyn_semigroup}
Let any $V\in\mathcal B(\mathcal H)$. The operator $(-i\hH-\hGA_V)$ on $\mathcal B^1(\mathcal H)$ with 
domain equal to that of $D(\hH)$ is the generator of a strongly continuous, positive, trace-preserving 
semigroup on $\mathcal B^1(\mathcal H)$, formally denoted by $(e^{-it\hH-t\hGA_V})_{t\in\mathbb R_+}$.
\end{lemma}
\begin{proof}
Apply \cite[Thm.~5.2]{Davies76} with $\mathscr J:= V(\cdot) V^\dagger $ so $\mathscr J^*(\mathbf 1)=V^
\dagger V$.
\end{proof}
\noindent
The corresponding semigroup is completely positive by its {\sc gksl} form, 
yet positivity and trace-preservation would suffice for what follows.

Thus even if $H_0$ is unbounded, for every choice of constant controls 
Eq.~\eqref{eq:GA} gives rise to a strongly continuous one-parameter semigroup $(e^{-it\hH-t\hGA_V})_{t\in\mathbb R_+}$
and therefore we obtain a well-defined reachable sets $\reach_{\Sigma_V}(\rho_0)$ in the sense of 
Eq.~\eqref{eq:reach_def} for all initial values $\rho_0\in\mathbb D(\mathcal H)$. The fact, that
the ODE \eqref{eq:GA} allows classical solutions only on a dense domain 
of initial values, may be neglect when specifying $\reach_{\Sigma_V}(\rho_0)$.

With the stage being set, we only need the Trotter product formula for contraction 
semigroups on Banach spaces before we can highlight how the proof of Theorem~\ref{thm_normal_V} 
changes to the new frame.

\begin{lemma}[\cite{ReedSimonII}, Thm.~X.51]\label{lemma_trotter_banach}
Let $\mathbf{A}_1$ and $\mathbf{A}_2$ be generators of contraction semigroups on 
$\mathcal B^1(\mathcal H)$, i.e.~strongly continuous semigroups of operator norm less or
equal one for all $t\in\mathbb R_+$. Suppose that the closure of $(\mathbf{A}_1+\mathbf{A}_2)$
generates a contraction semigroup on $\mathcal B^1(\mathcal H)$ and denote by 
$\overline{(\mathbf{A}_1+\mathbf{A}_2)}$ its closure. Then for all $\rho\in\mathcal B^1(\mathcal H)$
and all (fixed) $t \geq 0$
$$
\lim_{n\to\infty}\Vert  (e^{-t\mathbf{A}_1/n}e^{-t\mathbf{A}_2/n})^n(\rho)-e^{-t\overline{(\mathbf{A}_1+\mathbf{A}_2)}}  (\rho)\Vert_1=0\,.
$$
\end{lemma}
\noindent
With all these ingredients we are prepared to
\begin{proof}[Generalizing the proof of Thm.~\ref{thm_normal_V} to unbounded $H_0$]
``$\subseteq$'': As $V$ is assumed to be normal, one has $\hGA_V(\mathbf{1})=0$ and so the 
corresponding one-parameter semigroup is in $\mathbb S(\mathcal H)$, i.e.~it consists of bi-stochastic 
quantum maps. To see  that $e^{-it\hH-t\hGA_V}(\rho)\prec\rho$ for all $\rho\in\mathbb D(\mathcal H)$
and $t\in\mathbb R_+$  we note
\begin{itemize}
\item $e^{-it\hH}(\rho)\prec\rho$ as unitary channels do not change the eigenvalues. Thus, majorization 
cannot increase if the noise $\hGA_V$ is switched off.
\item $e^{-t\hGA_V}(\rho)\prec\rho$ by Lemma \ref{lemma_li}.
\end{itemize}
Therefore and since $\prec$ is a preorder (so in particular transitive), one has
\begin{equation}\label{eq:maj_sequence}
(e^{-it\hH/n}e^{-t\hGA_V/n})^n(\rho)\in \lbrace \omega\in\mathbb D(\mathcal H)\,|\,\omega\prec\rho\rbrace
\end{equation}
for all $n\in
\mathbb N_0$.
Now apply Lemma \ref{lemma_trotter_banach} to conclude that $(e^{-it\hH/n}e^{-t\hGA_V/n})^n(\rho)$ 
converges to $e^{-it\hH-t\hGA_V}(\rho)$ in trace norm for all $t\in \mathbb R_+$ and $n \to \infty$, Then, 
by Eq.~\eqref{eq:maj_sequence} in combination with Lemma \ref{lemma_maj_closed} (the set of majorized states is 
trace-norm closed), we conclude $e^{-it\hH-t\hGA_V}(\rho)\prec \rho$.

We saw earlier that $(e^{-it\hH})_{t\in\mathbb R_+}$ is a contractive semigroup. The same holds for 
$(e^{-t\hGA_V})_{t\in\mathbb R_+}$ by Lemma \ref{lemma_qc_norm} as well as $(e^{-it\hH-t\hGA_V})_{t\in
\mathbb R_+}$ by Lemma \ref{lemma_open_dyn_semigroup} \& \cite[Prop.~2]{vE_dirr_semigroups}
\footnote{The proof only uses positivity and trace-preservation, so we apply the respective 
result.}. The respective generators are all densely defined (on at least $D(\hH)$) and closed, 
so Lemma \ref{lemma_trotter_banach} yields
$$
\lim_{n\to\infty}\Vert  (e^{-it\hH/n}e^{-t\hGA_V/n})^n(\rho)-e^{-it\hH-t\hGA_V} (\rho)\Vert_1=0
$$
for all $\rho\in\mathcal B^1(\mathcal H)$, which shows the inclusion in question.

\medskip

``$\supseteq$'': Generalizing this inclusion to unbounded $H_0$ is easier as 
we only have to modify Lemma \ref{rem_trotter}. By 
the same line of reasoning, Lemma \ref{lemma_trotter_banach} gives
$$
\lim_{n\to\infty}\Vert  (e^{it\hH/n}e^{-(it\hH+t\hGA_V)/n} )^n(\rho)-e^{-t\hGA_V}(\rho)\Vert_1=0
$$
for all $\rho\in\mathbb D(\mathcal H)$ and $t\geq 0$. With this, the original proof 
of ``$\supseteq$'' holds without further changes---since we have never used the Trotter product formula 
explicitly in the uniform or norm topology, but only in the strong topology, i.e.~when applied to 
some density or trace-class operator, see also Lemma \ref{lemma_trotter_approx}.

Thereby all instances of Theorem~\ref{thm_normal_V} are finally proven.
\end{proof}

\bibliographystyle{mystyle}

\begin{thebibliography}{10}

\bibitem{AlbUhlm82}
Alberti, P. and Uhlmann, A., \emph{{Stochasticity and Partial Order: Doubly
  Stochastic Maps and Unitary Mixing}}, Reidel, Boston, 1982.

\bibitem{Ando89}
Ando, T., \emph{Lin. Alg. Appl.} \textbf{118} (1989), 163.

\bibitem{berberian1976}
Berberian, S., \emph{{Introduction to Hilbert Space}}, Amer. Math. Soc.,
  Chelsea, 1976.

\bibitem{BSH16}
Bergholm, V., Wilhelm, F., and Schulte-Herbr{\"u}ggen, T., \emph{{Arbitrary
  $n$-Qubit State Transfer Implemented by Coherent Control and Simplest
  Switchable Local Noise}}, 2016, \url{https://arxiv.org/abs/1605.06473v2}.

\bibitem{boscain2012weak}
Boscain, U., Caponigro, M., Chambrion, T., and Sigalotti, M., \emph{Commun.
  Math. Phys.} \textbf{311} (2012), 423.

\bibitem{Bro72}
Brockett, R.~W., \emph{SIAM J. Control} \textbf{10} (1972), 265.

\bibitem{Bro73}
Brockett, R.~W., \emph{SIAM J. Appl. Math.} \textbf{25} (1973), 213.

\bibitem{caponigro2018exact}
Caponigro, M. and Sigalotti, M., \emph{SIAM J. Control Optim.} \textbf{56}
  (2018), 2901.

\bibitem{CarboneFagnola01}
Carbone, R. and Fagnola, F., in: \emph{Proceedings of the Conference Quantum
  Probability and Infinite-Dimensional Analysis}, 57--76, World Sci. Publ.,
  Singapore, 2003.

\bibitem{Mart14}
Chen, Y., Neill, C., Roushan, P., Leung, N., Fang, M., Barends, R., Kelly, J.,
  Campbell, B., Chen, Z., Chiaro, B., Dunsworth, A., Jeffrey, E., Megrant, A.,
  Mutus, J.~Y., O’Malley, P.~J.~J., Quintana, C.~M., Sank, D., Vainsencher,
  A., Wenner, J., White, T.~C., Geller, M.~R., Cleland, A.~N., and Martinis,
  J.~M., \emph{Phys. Rev. Lett} \textbf{113} (2014), 220502.

\bibitem{TSING-96}
Cheung, W.-S. and Tsing, N.-K., \emph{Lin.~Multilin.~Alg.} \textbf{41} (1996),
  245.

\bibitem{Davies76}
Davies, E.~B., \emph{{Quantum Theory of Open Systems}}, Academic Press, London,
  1976.

\bibitem{DiHeGAMM08}
Dirr, G. and Helmke, U., \emph{GAMM-Mitteilungen} \textbf{31} (2008), 59.

\bibitem{DHKS08}
Dirr, G., Helmke, U., Kurniawan, I., and Schulte-Herbr{\"u}ggen, T., \emph{Rep.
  Math. Phys.} \textbf{64} (2009), 93.

\bibitem{dirr_ve}
Dirr, G. and vom Ende, F., \emph{Lin. Multilin. Alg.}  (2018), in press,
  \url{https://doi.org/10.1080/03081087.2018.1515884} and addendum
  \url{https://doi.org/10.1080/03081087.2019.1604624}.

\bibitem{CDC19}
Dirr, G., vom Ende, F., and Schulte-Herbr{\"u}ggen, T., \emph{{Reachable Sets
  from Toy Models to Controlled Markovian Quantum Systems}}, 2019,
  \url{https://arxiv.org/abs/1905.01224}.

\bibitem{DongPetersen2010}
Dong, D. and Petersen, I., \emph{IET Control Theory Appl.} \textbf{4} (2010),
  2651 .

\bibitem{Elliott09}
Elliott, D., \emph{{Bilinear Control Systems: Matrices in Action}}, Springer,
  London, 2009.

\bibitem{Fan49}
Fan, K., \emph{Proc. Natl. Acad. Sci. USA} \textbf{35} (1949), 652.

\bibitem{gohberg66}
Gohberg, I. and Markus, A., \emph{Amer. Math. Soc. Transl. Ser. 2} \textbf{52}
  (1966), 201 .

\bibitem{GKS76}
Gorini, V., Kossakowski, A., and Sudarshan, E., \emph{J. Math. Phys.}
  \textbf{17} (1976), 821.

\bibitem{Haroche13}
Haroche, S., \emph{Ann. Phys.} \textbf{525} (2013), 753.

\bibitem{JS72}
Jurdjevic, V. and Sussmann, H., \emph{J. Diff. Equat.} \textbf{12} (1972), 313.

\bibitem{keyl18InfLie}
Keyl, M., \emph{{Quantum Control in Infinite Dimensions and Banach-Lie
  Algebras: Pure Point Spectrum}}, arXiv:1812.09211, 2018.

\bibitem{Koss72}
Kossakowski, A., \emph{Bull. Acad. Pol. Sci., Ser. Sci. Math. Astron. Phys.}
  \textbf{20} (1972), 1021.

\bibitem{Koss72b}
Kossakowski, A., \emph{Rep. Math. Phys.} \textbf{3} (1972), 247.

\bibitem{Kraus83}
Kraus, K., \emph{{States, Effects, and Operations}}, Lecture Notes in Physics,
  Vol.~190, Springer, Berlin, 1983.

\bibitem{Li94}
Li, C.-K., \emph{Lin. Multilin. Alg.} \textbf{37} (1994), 51.

\bibitem{Li13}
Li, Y. and Busch, P., \emph{J. Math. Anal. Appl.} \textbf{408} (2013), 384 .

\bibitem{Lind76}
Lindblad, G., \emph{Commun. Math. Phys.} \textbf{48} (1976), 119.

\bibitem{MarshallOlkin}
Marshall, A., Olkin, I., and Arnold, B., \emph{{Inequalities: Theory of
  Majorization and Its Applications}}, Springer, New York, 2011, second ed.

\bibitem{MeiseVogt97en}
Meise, R. and Vogt, D., \emph{{Introduction to Functional Analysis}}, Oxford
  Graduate Texts in Mathematics, Oxford University Press, Oxford, 1997.

\bibitem{mirrahimi2004controllability}
Mirrahimi, M. and Rouchon, P., \emph{IEEE Trans. Automatic Control (IEEE-TAC)}
  \textbf{49} (2004), 745.

\bibitem{nadler1978}
Nadler, S., \emph{{Hyperspaces of Sets: A Text with Research Questions}}, M.
  Dekker, 1978.

\bibitem{ODS11}
O'Meara, C., Dirr, G., and Schulte-Herbr{\"u}ggen, T., \emph{IEEE Trans. Autom.
  Contr.} \textbf{57} (2012), 2050.

\bibitem{ReedSimonI}
Reed, M. and Simon, B., \emph{{Methods of Modern Mathematical Physics. Vol.~I:
  Functional Analysis}}, Academic Press, San Diego, 1980.

\bibitem{ReedSimonII}
Reed, M. and Simon, B., \emph{{Methods of Modern Mathematical Physics. Vol.~II:
  Fourier Analysis, Self-Adjointness}}, Academic Press, San Diego, 1975.

\bibitem{Haroche11}
Sayrin, C., Dotsenko, I., Zhou, X., Peaudecerf, B., Rybarczyk, T., Gleyzes, S.,
  Rouchon, P., Mirrahimi, M., Amini, H., Brune, M., Raimond, J., and Haroche,
  S., \emph{Nature} \textbf{477} (2011), 73.

\bibitem{OSID17}
Schulte-Herbr{\"u}ggen, T., Dirr, G., and Zeier, R., \emph{Open Syst. Quant.
  Information Dyn.} \textbf{24} (2017), 1740019.

\bibitem{OSID_Werner_17}
Siemon, I., Holevo, A., and Werner, R., \emph{Open Syst. Quant. Information
  Dyn.} \textbf{24} (2017), 1740015.

\bibitem{Sontag}
Sontag, E., \emph{{Mathematical Control Theory}}, Springer, New York, 1998, second~ed.

\bibitem{SJ72}
Sussmann, H. and Jurdjevic, V., \emph{J. Diff. Equat.} \textbf{12} (1972), 95.

\bibitem{Uhlm71}
Uhlmann, A., \emph{{Wiss.~Z.~Karl-Marx-Univ.~Leipzig, Math.~Nat.~R.}}
  \textbf{20} (1971), 633.

\bibitem{vE_dirr_semigroups}
vom Ende, F. and Dirr, G., \emph{{Unitary Dilations of Discrete
  Quantum-Dynamical Semigroups}}, arXiv:1804.00918, 2018.

\bibitem{NEUM-37}
von Neumann, J., \emph{Tomsk Univ. Rev.} \textbf{1} (1937), 286, [reproduced
  in: {\em John von Neumann: Collected Works}, A.H. Taub, Ed., Vol. IV:
  Continuous Geometry and Other Topics, Pergamon Press, Oxford, 1962, pp
  205-219].

\bibitem{Wolf08a}
Wolf, M.~M. and Cirac, J.~I., \emph{Commun. Math. Phys.} \textbf{279} (2008),
  147.

\bibitem{Yuan10}
Yuan, H., \emph{IEEE. Trans. Autom. Contr.} \textbf{55} (2010), 955.

\end{thebibliography}

\end{document}